\documentclass[journal]{IEEEtran}
\IEEEoverridecommandlockouts
\usepackage{amsmath,amsfonts,amssymb,euscript, graphicx, 
epsfig,
enumerate,float,afterpage, subfigure, ifthen}%

\newtheorem{thm}{Theorem}

\newtheorem{lem}{Lemma}
\newtheorem{defn}{Definition}
\newtheorem{exer}{Exercise}

\newcommand{\expect}[1]{\mathbb{E}\left\{#1\right\}}
\newcommand{\defequiv}{\mbox{\raisebox{-.3ex}{$\overset{\vartriangle}{=}$}}}

\newcommand{\bv}[1]{{\boldsymbol{#1} }}
\newcommand{\script}[1]{{{\cal{#1} }}}

\begin{document}

\title
  {Stability and Capacity Regions for Discrete Time Queueing Networks}
\author{Michael J. Neely
\thanks{Michael J. Neely is with the  Electrical Engineering department at the University
of Southern California, Los Angeles, CA. (web: http://www-rcf.usc.edu/$\sim$mjneely).} 
\thanks{This material is supported in part  by one or more of 
the following: the DARPA IT-MANET program
grant W911NF-07-0028, 
the NSF Career grant CCF-0747525, and continuing through participation in the 
Network Science Collaborative Technology Alliance sponsored
by the U.S. Army Research Laboratory.}}

\markboth{}{Neely}

\maketitle

\begin{abstract}   
We consider stability and network capacity in discrete time queueing
systems. Relationships between four common notions of stability are
described.  Specifically, we consider rate stability, mean rate stability, 
steady state stability, and strong stability. 
We then consider networks of queues with random events
and control actions that can be implemented over time to affect arrivals 
and service at the queues.  The control actions also generate a vector of additional
\emph{network attributes}.  We characterize the \emph{network capacity region}, 
being the closure of the set of all rate vectors that can be supported subject to network
stability and to additional time average attribute constraints.   
We show that (under mild technical assumptions) the capacity region is the
same under all four stability definitions.  Our capacity achievability proof 
uses the drift-plus-penalty method of Lyapunov optimization, and provides
full details for the case when network states obey a \emph{decaying memory property}, 
which holds for finite state ergodic systems and more general systems. 
\end{abstract} 


\section{Introduction} 

This paper considers stability and network capacity in discrete time 
queueing systems.   These issues arise in the analysis and control 
of modern data networks, including wireless and ad-hoc mobile networks, 
and are also important in many other application areas.  We have found
that many researchers have questions about the relationships between 
the different types of stability that can be used for network analysis. 
This paper is written to address those questions by  providing  details on
stability that are mentioned in other papers but are not proven due to lack of 
space. 

We consider the four most common types of stability from the literature:  rate 
stability, mean rate stability, steady state stability, and strong  stability.  
We first show that, under mild technical assumptions, strong stability 
implies the other three, and hence can be viewed as the strongest definition
among the four.  Conversely, we show  that 
mean rate stability is the weakest definition, in that (under mild technical assumptions)
it is implied by the other three.   We also briefly describe additional stability definitions, 
such as existence of a steady state workload distribution as in 
 \cite{baccelli-book}\cite{queue-mazumdar}\cite{bertsekas-data-nets} (often analyzed
 with Markov chain theory and Lyapunov drift theory  
 \cite{asmussen-prob}\cite{foss-stability}\cite{kumar-meyn-stability}\cite{now} and/or
 fluid models \cite{dai-fluid}\cite{dai-balaji-lyap}), and
 discuss their relationships to the main four.

We then consider control for a 
general stochastic multi-queue network.  The network operates in discrete time with 
timeslots $t \in\{0, 1, 2, \ldots\}$.  Control actions are made on 
each slot in reaction to random network events, such as random traffic arrivals
or channel conditions.  The control actions and network events affect arrivals and service
in the queues, and also generate a vector of \emph{network attributes}, being additional
penalties or rewards associated with the network (such as power expenditures, packet drops, etc.). 
We assume the system satisfies the mild technical assumptions 
needed for the above stability results.  We further assume that network events
have a \emph{decaying memory property}, 
a property typically exhibited by finite state 
ergodic systems as well as more general
systems. 
As in \cite{now}, 
we define the \emph{network capacity region} $\Lambda$
as the closure of the set of all traffic rate vectors that can be supported 
subject to network stability and to an additional set of time average attribute constraints. 
 We show that if traffic rates are outside
of the set $\Lambda$, then under any algorithm there must be at least one queue that is 
not mean rate stable.  Because mean rate stability is the \emph{weakest} definition, it 
follows that 
the network cannot be stable under any of the four definitions if traffic rates are outside of $\Lambda$. 
Conversely, we show that if the traffic rate vector is an  interior point of the set 
$\Lambda$, then it is possible to design 
an algorithm that makes all queues strongly stable (and hence it is also possible to achieve stability for the other three stability definitions). 
Because the capacity region is defined as a closure, it follows that it is invariant under any of
these four stability definitions.  

As an example, consider a simple discrete time $GI/GI/1$ queue with fixed size packets and
arrivals $a(t)$ that
are i.i.d. over slots 
with $\expect{a(t)} = \lambda$ packets/slot,  and independent 
time varying service rates $b(t)$ that 
are i.i.d. over slots with $\expect{b(t)} = 1/2$ packets/slot.   The ``mild technical assumptions''
that we impose here are that the second moments of the $a(t)$ and $b(t)$ processes are
finite.  In this setting, 
the capacity region is the
set of all arrival rates $\lambda$ such that $0 \leq \lambda \leq 1/2$.  It turns out that
the queue is rate stable (and mean rate stable)  if and only if $\lambda \leq 1/2$.
However, for steady state stability and strong stability we typically require $\lambda < 1/2$ (with
an exception in certain deterministic special cases).  Thus, the set of all rates for which the
queue is stable differs only by one point (the point $\lambda = 1/2$) under the four different
definitions of stability, and the closure of this set is identical for all four. 
There are indeed alternative (problematic) definitions of ``stability'' 
that would give rise to a different capacity region, although these 
are typically not used for networks and are not considered here 
(see an example in \cite{neely-thesis} of a problematic 
definition that says a queue is ``empty-stable'' if it empties infinitely often, and see Section \ref{section:problematic} for another problematic example).  

The above 1-queue example considers $a(t)$ and $b(t)$ processes that are i.i.d. over slots. 
However, our stability and capacity region analysis is more general and is presented in
terms of processes that are possibly non-i.i.d. over slots, assuming only that they have well
defined time averages with a mild ``decaying memory''
property.  We show that the capacity region is achievable using a strong drift-plus-penalty
method that we derived in previous papers \cite{neely-thesis}\cite{now}\cite{neely-energy-it}\cite{neely-fairness-infocom05}.    This method
treats joint stability and performance optimization, and extends the pioneering work on network 
stability in \cite{tass-radio-nets}\cite{tass-server-allocation}.
 These results are 
easier to derive in a context when network arrival and channel vectors are i.i.d. over slots. 
The prior work on network stability in \cite{tass-one-hop} treats general Markov-modulated channels.
Arrivals and channels with a ``decaying memory property'' are treated for stability in 
\cite{neely-power-network-jsac}\cite{neely-thesis}. 
Joint stability and utility optimization are considered for non-i.i.d. models
in 
\cite{rahul-cognitive-tmc}\cite{longbo-profit-allerton07}\cite{neely-maximal-bursty-ton}\cite{neely-mesh} 
for different types of networks.  This 
paper provides the full details for the non-i.i.d. case of joint stability and utility 
optimization for general networks with general time average
constraints.  
A more general ``universal scheduling'' model
is treated in \cite{neely-universal-scheduling}, which uses sample path analysis and 
considers possibly non-ergodic systems with no probability model, although the fundamental 
concept of a ``capacity region'' does not make sense in such a  context. 

\section{Queues}

Let $Q(t)$ represent the contents of a single discrete time queueing system defined over
integer timeslots $t \in \{0, 1, 2, \ldots\}$.   Specifically, the initial state $Q(0)$ is assumed to be a non-negative  
real valued random variable. 
Future states are driven by stochastic arrival and server processes $a(t)$ and $b(t)$ 
according to the following dynamic equation: 
\begin{equation} \label{eq:q-dynamics} 
 Q(t+1) = \max[Q(t) - b(t), 0] + a(t) \: \: \mbox{ for $t \in \{0, 1, 2, \ldots\}$} 
 \end{equation} 
 We call $Q(t)$ the \emph{backlog} on slot $t$, as it can represent an
 amount of work that needs to be done. 
The stochastic processes $\{a(t)\}_{t=0}^{\infty}$ and $\{b(t)\}_{t=0}^{\infty}$ are sequences of real valued
random variables defined over slots $t \in \{0, 1, 2, \ldots\}$.  

The value of 
$a(t)$  represents the amount of new work that arrives on slot $t$, and is assumed to be non-negative. 
The value of $b(t)$ represents the amount of work the server of the queue can process on slot $t$.  For most
physical queueing systems, $b(t)$ is assumed to be non-negative, although it is sometimes 
convenient to allow $b(t)$ to take negative values.  This is useful for the 
\emph{virtual queues} defined in future sections, where $b(t)$ can be interpreted as a (possibly negative) attribute.\footnote{Assuming that
the $b(t)$ value in (\ref{eq:q-dynamics}) is
possibly negative also allows treatment of modified queueing models that 
place new arrivals inside the $\max[\cdot, 0]$
operator.  For example, a queue with dynamics $\hat{Q}(t+1) = \max[\hat{Q}(t) - \beta(t) + \alpha(t), 0]$ is the 
same as (\ref{eq:q-dynamics}) with $a(t) = 0$ and $b(t) = \beta(t) - \alpha(t)$ for all $t$. Leaving $a(t)$ outside
the $\max[\cdot, 0]$ is crucial for treatment of multi-hop networks, where $a(t)$ can be a sum of exogenous and endogenous
arrivals.} 
Because we assume $Q(0) \geq 0$ and $a(t) \geq 0$ for all slots $t$, it is clear from (\ref{eq:q-dynamics}) 
that $Q(t) \geq 0$ for all slots $t$. 

The units of $Q(t)$, $a(t)$, and $b(t)$ depend on the context of the system.  For example, in a communication
system with fixed size data units, these quantities might be integers with units of \emph{packets}.  Alternatively, 
they might be real numbers with units of \emph{bits}, \emph{kilobits}, or some other unit of unfinished work 
relevant to the system.  

We can equivalently re-write the dynamics (\ref{eq:q-dynamics}) without the non-linear $\max[\cdot, 0]$ operator
as follows: 
\begin{equation} \label{eq:q-dynamics-tilde} 
 Q(t+1) = Q(t) - \tilde{b}(t) + a(t)  \: \: \mbox{ for $t \in \{0, 1, 2, \ldots\}$} 
 \end{equation} 
where $\tilde{b}(t)$ is the actual work processed on slot $t$ (which may be less than the offered
amount $b(t)$ if there is little or no backlog in the system at slot $t$).  Specifically, $\tilde{b}(t)$ 
is mathematically defined: 
\[ \tilde{b}(t) \defequiv \min[b(t), Q(t)] \]
Note by definition that $\tilde{b}(t) \leq b(t)$ for all $t$.  
The dynamic equation (\ref{eq:q-dynamics-tilde}) yields a 
simple but important property for all sample paths, described in the following lemma. 

\begin{lem} (Sample Path Property) 
For any discrete time queueing system  described by 
(\ref{eq:q-dynamics}),  and for any two 
slots $t_1$ and $t_2$ such that $0 \leq t_1 < t_2$, we have: 
\begin{eqnarray} 
Q(t_2) - Q(t_1) &=& \sum_{\tau=t_1}^{t_2-1} a(\tau) - \sum_{\tau=t_1}^{t_2-1} \tilde{b}(\tau) \label{eq:io-tilde} \\
Q(t_2) - Q(t_1) &\geq& \sum_{\tau=t_1}^{t_2-1} a(\tau) - \sum_{\tau=t_1}^{t_2-1} b(\tau) \label{eq:io-b2} 
\end{eqnarray} 
Therefore, for any $t>0$ we have:  
\begin{eqnarray} 
 \frac{Q(t)}{t} - \frac{Q(0)}{t} &=& \frac{1}{t} \sum_{\tau=0}^{t-1} a(\tau) - \frac{1}{t}\sum_{\tau=0}^{t-1} \tilde{b}(\tau) \label{eq:illuminate} \\
 \frac{Q(t)}{t} - \frac{Q(0)}{t} &\geq& \frac{1}{t} \sum_{\tau=0}^{t-1} a(\tau) - \frac{1}{t}\sum_{\tau=0}^{t-1} b(\tau) \label{eq:illuminate2} 
 \end{eqnarray}  
\end{lem} 
\begin{proof} 
By (\ref{eq:q-dynamics-tilde}) we have for any slot  $\tau\geq 0$: 
\[ Q(\tau+1) - Q(\tau) = a(\tau) - \tilde{b}(\tau) \]
Summing the above over $\tau \in \{t_1, \ldots, t_2 -1\}$ and using  telescoping
sums yields: 
\[ Q(t_2) - Q(t_1) =\sum_{\tau=t_1}^{t_2-1} a(\tau)  - \sum_{\tau=t_1}^{t_2-1} \tilde{b}(\tau) \]
This proves (\ref{eq:io-tilde}). Inequality (\ref{eq:io-b2}) follows because $\tilde{b}(\tau) \leq b(\tau)$
for all $\tau$. Inequalities  (\ref{eq:illuminate}) and (\ref{eq:illuminate2}) 
follow by substituting $t_1 =0$, $t_2 = t$,
and dividing by $t$. 
\end{proof} 

The equality (\ref{eq:illuminate}) is illuminating.  It shows that
$Q(t)/t\rightarrow0$ as $t \rightarrow \infty$ 
if and only if the time average of the process  $a(t) - \tilde{b}(t)$ is zero 
(where the time average of $a(t) - \tilde{b}(t)$ 
is the limit of the right hand side of (\ref{eq:illuminate})). 
This happens when  the time average rate of arrivals $a(t)$ 
is equal to the time average rate of actual 
departures $\tilde{b}(t)$. This motivates the 
definitions of \emph{rate stability} and \emph{mean rate stability}, defined in the next section. 

\section{Rate Stability} 
Let $Q(t)$ be the backlog process in a discrete time queue.  We assume only that 
$Q(t)$ is non-negative and evolves over slots
$t \in \{0, 1, 2, \ldots\}$ according to some probability law.\footnote{All of our stability definitions
can be extended to treat discrete time stochastic processes $Q(t)$ 
that can possibly be negative by substituting $|Q(t)|$ into the definitions, which is sometimes
useful in contexts (not treated here) where virtual queues can be possibly negative, as 
in \cite{neely-universal-scheduling}\cite{neely-mwl-ita}.}  

\begin{defn}   \label{def:rate-stable} A discrete time queue $Q(t)$ is \emph{rate stable} if: 
\[ \lim_{t\rightarrow\infty} \frac{Q(t)}{t} = 0 \: \: \mbox{ with probability 1} \]
\end{defn}  

\begin{defn}  \label{def:mean-rate-stable} A discrete time queue $Q(t)$ is \emph{mean rate stable} if: 
\[ \lim_{t\rightarrow\infty} \frac{\expect{Q(t)}}{t} = 0 \]
\end{defn}

%
Neither rate stability nor mean rate stability implies the other (see counter-examples
given in Section \ref{section:counterexamples}). However, rate
stability implies mean rate stability under the following mild technical assumptions.  

\begin{thm} \label{thm:rs-implies-mrs} 
(Rate Stability \& Bounding Assumptions Implies Mean Rate Stability) 
Consider a queue $Q(t)$ with dynamics (\ref{eq:q-dynamics}), with $b(t)$ real valued
(possibly negative) and $a(t)$ non-negative.  Suppose that $Q(t)$ is rate stable. 

a) Suppose there are finite constants $\epsilon>0$ and $C>0$ such that $\expect{(a(t)+b^{-}(t))^{1+\epsilon}} \leq C$ for all $t$, where $b^-(t)$ is defined: 
\begin{equation} \label{eq:b-minus} 
  b^-(t) \defequiv -\min[b(t), 0]  
  \end{equation} 
Then $Q(t)$ is mean rate stable. 

b) Suppose  there is a non-negative random variable $Y$ with $\expect{Y}  < \infty$ and
such that for all $t \in \{0, 1, 2, \ldots\}$
we have:  
\begin{eqnarray}
\expect{a(t)+b^-(t)|a(t)+b^-(t)>y}Pr[a(t)+b^-(t)>y]  \nonumber \\
\leq \expect{Y|Y>y}Pr[Y>y]  \: \: \forall y \in \mathbb{R}  \label{eq:stoch-greater} 
\end{eqnarray} 
Then $Q(t)$ is mean rate stable.
\end{thm}
\begin{proof} 
See Appendix A. 
\end{proof}

We note that the condition (\ref{eq:stoch-greater}) holds whenever 
the random variable $Y$ is \emph{stochastically greater than or equal to} 
the random variable $a(t) + b^-(t)$ for all $t$ \cite{ross}.  This condition also trivially
holds whenever $a(t) + b^-(t)$ is stationary, having the same probability distribution for all slots $t$
but not necessarily being i.i.d. over slots, and satisfies $\expect{a(0) + b^-(0)} < \infty$.  
This is because, in this stationary case, 
we can use $Y = a(0) + b^-(0)$.  The condition in part (a) does not require stationarity, but  
requires
a uniform bound on the ``$(1+\epsilon)$'' moment for some $\epsilon>0$.  This certainly 
holds whenever the second moments of 
$a(t) + b^{-}(t)$ are bounded by some finite constant $C$ for all $t$
(so that $\epsilon=1$), as assumed in our network analysis of Section \ref{section:network}. 

The next theorem gives intuition on rate stability and mean rate stability for queues with 
well defined time average arrival and server rates. 

\begin{thm}  \label{thm:rate-stability} (Rate Stability Theorem) Suppose $Q(t)$ evolves according to (\ref{eq:q-dynamics}), 
with $a(t) \geq 0$ for all $t$, and with $b(t)$ real valued (and possibly negative) for all $t$. 
Suppose that the time averages of the processes $a(t)$ and $b(t)$  converge with probability $1$ to finite constants
$a_{av}$ and $b_{av}$,  so that: 
\begin{eqnarray} 
\lim_{t\rightarrow\infty} \frac{1}{t}\sum_{\tau=0}^{t-1} a(\tau) = a_{av} &  \mbox{ with probability $1$} \label{eq:Aav} \\
\lim_{t\rightarrow\infty} \frac{1}{t}\sum_{\tau=0}^{t-1} b(\tau) = b_{av} & \mbox{ with probability $1$} \label{eq:muav} 
\end{eqnarray} 
Then: 

(a)  $Q(t)$ is rate stable if and only if $a_{av} \leq b_{av}$.  

(b) If $a_{av} > b_{av}$, then: 
\[ \lim_{t\rightarrow\infty} \frac{Q(t)}{t} = a_{av} - b_{av} \: \: \mbox{ with probability 1} \]

(c) Suppose there are finite constants $\epsilon>0$ and $C>0$ such 
that $\expect{(a(t)+b^{-}(t))^{1+\epsilon}} \leq C$ for all $t$, where $b^-(t)$ is defined
in (\ref{eq:b-minus}). 
Then 
$Q(t)$ is mean rate stable if and only if $a_{av} \leq b_{av}$.  

(d) Suppose there is a non-negative random variable $Y$ with $\expect{Y}  < \infty$ and
such that condition (\ref{eq:stoch-greater}) holds for all $t \in \{0, 1, 2, \ldots\}$. 
Then $Q(t)$ is mean rate stable if and only if $a_{av} \leq b_{av}$. 
\end{thm}

\begin{proof} (Theorem \ref{thm:rate-stability}) 
Suppose that $Q(t)$ is rate stable, so that $Q(t)/t \rightarrow 0$ with probability $1$. 
Because   (\ref{eq:illuminate2}) holds
for all slots $t>0$, we can take limits in (\ref{eq:illuminate2}) as $t \rightarrow \infty$ and use
 (\ref{eq:Aav})-(\ref{eq:muav})  to conclude that
 $0 \geq a_{av} - b_{av}$. Thus,  
$a_{av} \leq b_{av}$ is \emph{necessary} for rate stability.  The proof for sufficiency in part (a) and the proof of  
part (b) are not obvious and are 
developed in Exercises \ref{ex:rate-stable} and \ref{ex:rate-stability-b} of 
Section \ref{section:exercise}. 

To prove parts (c) and (d), 
suppose that $a_{av} \leq b_{av}$.  We thus know by part (a) 
that  $Q(t)$ is rate stable.  The conditions in parts (c) and (d) of this theorem 
correspond to the conditions given in Theorem \ref{thm:rs-implies-mrs}, and hence
$Q(t)$ is mean rate stable.

Now suppose that $a_{av} > b_{av}$.  It follows by part (b) that: 
\[ \lim_{t\rightarrow\infty} \frac{Q(t)}{t} = a_{av} - b_{av} \: \: \mbox{ with prob. 1} \]
Define $\delta \defequiv (a_{av} - b_{av})/2$. Note that: 
\[ \lim_{t\rightarrow\infty} Pr[Q(t)/t > \delta] = 1 \]
Therefore: 
\begin{eqnarray*}
 \expect{\frac{Q(t)}{t}} &\geq& \expect{\frac{Q(t)}{t}|\frac{Q(t)}{t} > \delta} Pr[Q(t)/t>\delta]  \\
&\geq& \delta Pr[Q(t)/t > \delta]
\end{eqnarray*}
Taking a limit yields: 
\[ \limsup_{t\rightarrow\infty} \expect{\frac{Q(t)}{t}} \geq \delta \]
and hence $Q(t)$ is not mean rate stable. 
\end{proof}

Prior sample path investigations
of constant service rate queues are provided in  
\cite{sample-path-queue}\cite{taha-paper}\cite{path-mazumdar}\cite{queue-mazumdar}, 
where it is shown that 
 rate stability holds whenever the arrival rate is strictly less than the service rate. 
Our proof of Theorem \ref{thm:rate-stability}(a) uses a different chain of
reasoning (developed in Exercises \ref{ex:rate-stable} and \ref{ex:rate-stability-b} of 
Section \ref{section:exercise}),
applies to queues with more general time
varying (and possibly negative)
service rates, and also shows the case $a_{av} = b_{av}$ ensures rate stability, which establishes
the simple necessary and sufficient condition $a_{av} \leq b_{av}$.

The assumption that $a(t)$ and $b(t)$ have well defined time averages $a_{av}$ and $b_{av}$
is crucial for the result of Theorem \ref{thm:rate-stability}.  One might intuitively 
suspect that if $a(t)$
has a well defined time average $a_{av}$, but $b(t)$ has $\limsup$ and $\liminf$ time
averages $b_{av}^{inf}$ and $b_{av}^{sup}$ such that $a_{av} < b_{av}^{inf} < b_{av}^{sup}$, then
$Q(t)$ is also rate stable.  This is not always true.  Thus, the existence of well defined
time averages provides enough structure to ensure queue sample paths are well behaved. 
The following theorem presents a more general necessary condition for rate stability that does
not require the arrival and server processes to have well defined time averages. 

\begin{thm}  \label{thm:gen-nec-rate} (Necessary Condition for Rate Stability) 
 Suppose $Q(t)$ 
 evolves according to (\ref{eq:q-dynamics}), 
 with any general processes $a(t)$ and $b(t)$ such that $a(t) \geq 0$ for all $t$.  Then: 
 
 (a) If $Q(t)$ is rate stable, then: 
 \begin{equation} \label{eq:limsup-rs} 
 \limsup_{t\rightarrow\infty} \frac{1}{t}\sum_{\tau=0}^{t-1} [a(\tau) - b(\tau)] \leq 0 \: \: \mbox{ with probability 1} 
 \end{equation} 
 
 (b) If $Q(t)$ is mean rate stable and if $\expect{Q(0)} < \infty$, then: 
 \begin{equation} \label{eq:limsup-mrs} 
 \limsup_{t\rightarrow\infty} \frac{1}{t}\sum_{\tau=0}^{t-1} \expect{a(\tau) - b(\tau)} \leq 0
 \end{equation} 
  \end{thm}  
 \begin{proof} 
The proof of (a) follows immediately by taking a 
$\limsup$ of both sides of (\ref{eq:illuminate2}) and noting that $Q(t)/t \rightarrow 0$ because $Q(t)$ is rate stable. 
The proof of (b) follows by first taking an expectation of (\ref{eq:illuminate2}) and then 
taking limits. 
 \end{proof}


\section{Stronger Forms of Stability}

 Rate stability and mean rate stability only describe the long term average rate of arrivals and departures
 from the queue, and do not say anything about the fraction of time the queue backlog exceeds
 a certain value, or about the time average expected backlog. 
 The stronger stability definitions given below are thus useful.   
 \begin{defn}  A discrete time queue $Q(t)$ is \emph{steady state stable} if: 
\[ \lim_{M\rightarrow\infty} g(M) = 0 \]
where for each $M\geq 0$,  
$g(M)$ is defined: 
\begin{equation} \label{eq:gm} 
 g(M) \defequiv \limsup_{t\rightarrow\infty} \frac{1}{t}\sum_{\tau=0}^{t-1} Pr[Q(\tau) > M] 
 \end{equation} 
\end{defn}  

\begin{defn}  A discrete time queue $Q(t)$ is \emph{strongly stable} if: 
\begin{equation} \label{eq:strongly-stable} 
 \limsup_{t\rightarrow\infty} \frac{1}{t}\sum_{\tau=0}^{t-1} \expect{Q(\tau)} < \infty 
 \end{equation} 
\end{defn}  

For discrete time ergodic Markov chains with countably infinite state space and with 
the property that, for each real value $M$,  the event $\{Q(t) \leq M\}$ corresponds to only a finite number of states,  
steady state stability implies the existence of a steady state distribution, and 
 strong stability implies  finite average backlog and (by Little's theorem \cite{bertsekas-data-nets})
 finite average delay.  Under mild boundedness assumptions, strong stability implies 
 all of the other forms of stability, as specified in Theorem \ref{thm:strong-stability} below.

 \begin{thm}  \label{thm:strong-stability} (Strong Stability Theorem) 
 Suppose $Q(t)$  evolves according to  (\ref{eq:q-dynamics}) for some
 general stochastic processes $\{a(t)\}_{t=0}^{\infty}$ and $\{b(t)\}_{t=0}^{\infty}$, where $a(t) \geq 0$ for 
 all $t$, and $b(t)$ is real valued for all $t$. Suppose $Q(t)$ is strongly stable.  Then: 
 
 (a) $Q(t)$ is steady state stable. 
 
 (b)  If there is a finite constant $C$ such that either $a(t) + b^{-}(t)  \leq C$ with probability 1 for all $t$
 (where $b^-(t)$ is defined in (\ref{eq:b-minus})),  
 or $b(t) - a(t) \leq C$ with probability 1 for all $t$, 
 then $Q(t)$ is rate stable, 
 so that $Q(t)/t \rightarrow 0$ with probability $1$.

 (c)   If there is a finite constant $C$ such that either $\expect{a(t) + b^{-}(t)} \leq C$ for all $t$,
 or $\expect{b(t) - a(t)} \leq C$ for all $t$, then $Q(t)$ is mean rate stable.  
 \end{thm}  
 
 \begin{proof} 
Part (a) is given in Exercise \ref{ex:strong-stability-implies-steady-state}. Part (c) is given 
in Appendix B, and part (b) is given in Appendix C.  
 \end{proof} 
 
 
 
 The above theorem shows that, under mild technical assumptions, strong stability implies
 all three other forms of stability.    Theorem \ref{thm:rs-implies-mrs} and Theorem \ref{thm:strong-stability}(c) show that (under mild technical assumptions) rate stability and strong stability both imply
 mean rate stability.  For completeness, the following theorem provides conditions under
 which steady state stability implies mean rate stability.   Collectively, these results can be 
 viewed as showing that strong stability is the \emph{strongest} definition of the four, and 
 mean rate stability is the \emph{weakest} definition of the four. 
 
 \begin{thm} \label{thm:ss-implies-mrs} Assume $Q(t)$ evolves according to (\ref{eq:q-dynamics})
 with $a(t)\geq 0$ and $b(t)$ real values for all $t$. Suppose that $Q(t)$ is steady state stable, 
 and that there is a finite constant $C$ such that $a(t) + b^-(t) \leq C$ with probability 1
 for all $t$.  Then $Q(t)$ is mean rate stable. 
 \end{thm}
 
 \begin{proof}
 See Appendix D. 
 \end{proof}
 
\subsection{Sample Path Versions of Stability} 

One might use a sample-path version of strong stability, saying that
a queue is \emph{sample-path strongly stable} if: 
\[ \limsup_{t\rightarrow\infty} \frac{1}{t}\sum_{\tau=0}^{t-1} Q(\tau) < \infty \: \: \mbox{ with prob. 1} \]
A sample-path version of steady-state stable would re-define the function $g(M)$ in (\ref{eq:gm})
by a function $h(M)$ as follows: 
\[ h(M) \defequiv \limsup_{t\rightarrow\infty} \frac{1}{t}\sum_{\tau=0}^{t-1} 1\{Q(\tau)>M\} \]
where $1\{Q(\tau)>M\}$ is an indicator function that is $1$ if $Q(\tau)>M$, and zero otherwise. 
We might say that the queue is \emph{sample-path steady-state stable} if
$\lim_{M\rightarrow\infty} h(M) = 0$ with probability 1.   These two additional stability definitions
are again implied by strong stability if one assumes the system has well defined limits (which is 
typically the case in systems defined on Markov chains), as shown in Appendix E. 

\subsection{A Problematic Stability Definition} \label{section:problematic} 

Finally, one might define another form of stability by requiring: 
\begin{equation} \label{eq:stronger} 
\limsup_{t\rightarrow\infty} \expect{Q(t)} < \infty 
\end{equation} 
 It is clear that if $\expect{Q(t)} < \infty$ for all $t$ and if the above holds, then the
 time average of $\expect{Q(t)}$ is also finite and so 
 $Q(t)$ is strongly stable. Hence, the  condition (\ref{eq:stronger}) 
 can be viewed as being 
 even ``stronger'' than strong stability.  Of course, in most systems defined over Markov
 chains, strong stability is equivalent to (\ref{eq:stronger}). 
 
 However, we do not consider (\ref{eq:stronger})  in our
 set of stability definitions for two reasons:  
 \begin{enumerate} 
 \item The strong stability definition that uses
 a time average is easier to work with, especially in Lyapunov drift arguments \cite{now}. 
 \item The condition (\ref{eq:stronger}) leads to a problematic counterexample if it
 were used as a definition of stability, as described below. 
 \end{enumerate} 
 
 To see why the definition (\ref{eq:stronger}) may be problematic, consider a simple
 discrete time ``Bernoulli/Bernoulli/1'' ($B/B/1$) queue, where arrivals $a(t)$ are i.i.d. over slots with $Pr[a(t) = 1] = \lambda$
 and $Pr[a(t) = 0] = 1-\lambda$, and the server process $b(t)$ 
 is independent and i.i.d. over slots with $Pr[b(t) = 1] = \mu$, $Pr[b(t) = 0] = 1-\mu$.  When 
 $\lambda < \mu$, it is easy to write the ergodic 
 birth-death Markov chain for this system, and one
 can easily derive that the Markov chain has a well defined steady state, steady state probabilities 
 decay geometrically, and average queue backlog and delay satisfy: 
 \[ \overline{Q} = \frac{\lambda(1-\lambda)}{\mu-\lambda}  \: \: , \: \: \overline{W} = \frac{1-\lambda}{\mu-\lambda} \]
 Thus, if $\lambda < \mu$, a good definition of stability would say that this system is stable. 
 
 Now suppose we take one particular sample path realization 
 of the $B/B/1$ queue, one for which time averages converge to the steady state values (which 
 happens with probability 1).  However, treat this sample path as given, so that all
 events are now \emph{deterministic}.  Thus, $a(t)$ and $b(t)$ are now deterministic functions. 
 Because we have not changed the actual sample path, 
 a good definition should also 
 say this deterministic variant is stable.  In this deterministic case, we have $\expect{Q(t)} = Q(t)$
 for all $t$, and so the expectation can grow arbitrarily large (since a $B/B/1$ queue can grow
 arbitrarily large).   Hence, for this deterministic example: 
 \[ \limsup_{t\rightarrow\infty} \expect{Q(t)} = \infty \]
 Thus, if we used (\ref{eq:stronger}) as a definition of stability, the random
 $B/B/1$ queue would be stable, but the deterministic version would not!
 Another way of saying this is that the original $B/B/1$ queue is stable, but if we 
 condition on knowing all future events, it becomes unstable!

 Because our definitions of rate stability, mean rate stability, sample path stability, and 
 strong stability incorporate time averages, these four forms of stability all say that both the random
 $B/B/1$ queue and its deterministic counterpart are stable.  
 Hence, the problem does not 
 arise for this example 
 in any of these four definitions. 
 
 Another problematic stability definition says that a queue is ``empty-stable'' if it
 empties infinitely often.   A discussion
 of why this is problematic is provided in \cite{neely-thesis}.

 \section{Counter-Examples}  \label{section:counterexamples} 

Here we provide counter-examples that show what can happen if the boundedness
assumptions of Theorems \ref{thm:rs-implies-mrs} and \ref{thm:strong-stability} are violated.

\subsection{Rate Stability Does Not Imply Mean Rate Stability}  \label{subsection:counterexamples-rate-mean}

Let $T$ be an integer random variable with a geometric distribution,  
such that $Pr[T>t] = 1/2^t$ for $t \in \{0, 1, 2, \ldots\}$.  Define $Q(t)$ over $t \in \{0, 1, 2, \ldots\}$ as follows: 
\[ Q(t) = \left\{ \begin{array}{ll}
                          2^{(2t)} &\mbox{ if $t < T$} \\
                             0  & \mbox{ otherwise} 
                            \end{array}
                                 \right.\]
It follows that $Q(t)/t \rightarrow 0$ with probability 1 (as eventually $t$ becomes larger than the random variable $T$). 
However: 
\[ \expect{Q(t)} = 2^{(2t)} Pr[t < T] = 2^{(2t)} 2^{-t} = 2^{t} \]
Therefore: 
\[ \lim_{t\rightarrow\infty}  \expect{Q(t)}/t = \lim_{t\rightarrow\infty} 2^t/t =  \infty \]
and hence $Q(t)$ is not mean rate stable.   

\subsection{Mean Rate Stability  Does Not Imply Rate Stability} 

Suppose that $Q(0)=0$, and that $Q(t)$  has independent values every slot $t$ for $t \in \{1, 2, 3, \ldots\}$, 
so that: 
\[ Q(t) = \left\{ \begin{array}{ll}
                          t &\mbox{ with probability $1/t$}  \\
                             0  & \mbox{ with probability $1- 1/t$} 
                            \end{array}
                                 \right.\]
Thus, for any time $t>0$ we have: 
\[ \expect{Q(t)} = t/t = 1 \]
It follows that $\expect{Q(t)}/t \rightarrow 0$, and so $Q(t)$ is mean rate stable. 

However, clearly $Q(t)/t = 1$ for any time $t$ such that $Q(t)>0$.  Because the probabilities of the independent 
events at which $Q(t)>0$ decay very slowly, it can be shown that there are an infinite number of times $t_n$ for which
$Q(t_n) > 0$.  Indeed, we have for any time $t>0$:
\[ Pr[\mbox{$Q(\tau) = 0$  for all $\tau\geq t$}] = \prod_{\tau=t}^{\infty} (1-1/\tau) = 0 \]
The infinite product can be shown to be zero for any $t>0$ 
by taking a $\log(\cdot)$ and showing that the resulting infinite sum is equal to 
$-\infty$. 
Therefore: 
\[ \limsup_{t\rightarrow\infty} \frac{Q(t)}{t} = 1 \: \: \mbox{ with probability 1} \]
and so the system is not rate stable.

\subsection{Strong Stability Neither Implies Mean Rate Stability Nor Rate Stability} 

Suppose that $Q(t) = t$ for $t \in \{1, 2, 4, 8, \ldots, 2^n, \ldots\}$ (i.e., for all timeslots $t$ that are powers of $2$). 
Suppose that $Q(t) = 0$ at all slots $t$ that are not powers of 2.  Because this process is deterministic, we have
$\expect{Q(t)} = Q(t)$, and for all $n \in \{0, 1, 2, \ldots\}$ we have: 
\[ \frac{1}{2^n+ 1} \sum_{\tau=0}^{2^n} Q(\tau) = \frac{1 + 2 + \ldots + 2^n}{2^n+1}  = \frac{2^{n+1} - 1}{2^n + 1} \]
The right hand side of the above expression converges to $2$ as $n \rightarrow \infty$. 
It can be shown that $\frac{1}{t}\sum_{\tau=0}^{t-1} Q(\tau)$ is the largest when sampled at the times in the above expression, 
and hence: 
\[ \limsup_{t\rightarrow\infty} \frac{1}{t}\sum_{\tau=0}^{t-1} Q(\tau) = 2 \]
Because $\expect{Q(\tau)} = Q(\tau)$ for all $\tau$, it follows that: 
\[ \limsup_{t\rightarrow\infty}\frac{1}{t}\sum_{\tau=0}^{t-1}\expect{Q(\tau)} = 2 \]
Therefore, $Q(t)$ is strongly stable. However, $Q(2^n)/2^n = 1$ for all $n \in \{1, 2, \ldots\}$, and so $Q(t)$ is not
rate stable or mean rate stable (rate stability and mean rate stability are equivalent when $Q(t)$ is deterministic). 
Note in this case that increases or decreases in  $Q(t)$ can be arbitrarily large, and hence this example 
does not satisfy the boundedness assumptions required in Theorem \ref{thm:strong-stability}.

\section{Network Scheduling} \label{section:network} 

Consider now the following 
multi-queue network model defined over discrete time $t\in\{0, 1, 2, \ldots\}$.  
There are $K$ queues
$\bv{Q}(t) = (Q_1(t), \ldots, Q_K(t))$, with dynamics for all $k \in \{1, \ldots, K\}$: 
\begin{equation} \label{eq:tilde-q}  
Q_k(t+1) = Q_k(t) - \tilde{b}_k(t) + \tilde{y}_k(t) + a_k(t) 
\end{equation} 
where $a_k(t)$ represents new exogenous arrivals to queue $k$ on slot $t$, 
$\tilde{b}_k(t)$ represents the actual amount served on slot $t$, and 
$\tilde{y}_k(t)$ represents additional arrivals.  The additional arrivals $\tilde{y}_k(t)$
may be due
to flow control operations (in which case we might have $a_k(t) = 0$ so that all new arrivals
are first passed through the flow control mechanism).  They might also be due to 
endogenous arrivals from other queues, which allows treatment of multi-hop networks.  

We assume the queue is always non-negative, as
are $\tilde{b}_k(t)$, $\tilde{y}_k(t)$, $a_k(t)$, and 
that the  $\tilde{b}_k(t)$ and $\tilde{y}_k(t)$ values respect the amount of data that is 
actually in each queue (not serving more or delivering more than the amount transferred
over the channel).  It is also useful to assume transmission decisions can be made independently
of queue backlog, and so we also define $b_k(t)$ and $y_k(t)$ for queue dynamics: 
\begin{equation} \label{eq:multi-q-dynamics} 
Q_k(t+1) \leq \max[Q_k(t) - b_k(t), 0] + y_k(t) + a_k(t) 
\end{equation} 
The inequality is due to the fact that the amount of actual new exogenous arrivals $\tilde{y}_k(t)$, being
a sum of service values in other queues that transmit to queue $k$, may not be as large
as $y_k(t)$ if these other queues have little or no backlog to transmit.  These values satisfy: 
\begin{eqnarray*}
0 \leq \tilde{y}_k(t) \leq y_k(t) & \forall k \in \{1, \ldots, K\}, \forall t \\
0 \leq \tilde{b}_k(t) \leq b_k(t) & \forall k \in \{1, \ldots, K\}, \forall t
\end{eqnarray*}

This model 
is similar to that given in \cite{now}\cite{neely-power-network-jsac}\cite{neely-thesis}, and
the capacity region  that we develop is also similar to prior work there.  

The network has a time varying \emph{network state} $\omega(t)$, 
possibly being a vector of channel conditions and/or additional random arrivals
for slot $t$.  A \emph{network control action} $\alpha(t)$ is chosen in reaction to the observed
network state $\omega(t)$ on slot $t$ (and possibly also in reaction to other network information,
such as queue backlogs), and takes values in some abstract set $\script{A}_{\omega(t)}$ that possibly
depends on $\omega(t)$. The $\omega(t)$ and $\alpha(t)$ values for slot $t$ affect arrivals and
service by: 
\begin{eqnarray*}
y_k(t) &=& \hat{y}_k(\alpha(t), \omega(t)) \: \: \forall k \in \{1, \ldots, K\} \\
b_k(t) &=& \hat{b}_k(\alpha(t), \omega(t)) \: \: \forall k \in \{1, \ldots, K\} 
\end{eqnarray*}
where $\hat{y}_k(\alpha(t), \omega(t))$ and $\hat{b}_k(\alpha(t), \omega(t))$ are 
general functions of $\alpha(t)$ and $\omega(t)$ (possibly non-convex and discontinuous). 
The $\omega(t)$ and $\alpha(t)$ values also affect an \emph{attribute vector} $\bv{x}(t) = (x_1(t), \ldots, x_M(t))$ for slot $t$, which can represent additional penalties or rewards associated with the network
states and control actions (such as power expenditures, packet admissions, packet drops, etc.). 
The components $x_m(t)$ can possibly be negative, and are 
general functions of $\alpha(t)$ and $\omega(t)$: 
\[ x_m(t) = \hat{x}_m(\alpha(t), \omega(t)) \]
 
 \subsection{Network Assumptions} 
 
 We assume that exogenous arrivals $a_k(t)$ satisfy: 
\begin{equation} \label{eq:lambda-k-assumption} 
 \lim_{t\rightarrow\infty} \frac{1}{t}\sum_{\tau=0}^{t-1} \expect{a_k(\tau)} = \lambda_k \: \: \forall k \in \{1, \ldots, K\} 
 \end{equation} 
where we call $\lambda_k$ the arrival rate for queue $k$. 
We assume the network state $\omega(t)$ is \emph{stationary} with a well defined stationary
distribution $\pi(\omega)$.  In the case when there are only a finite or countably infinite 
number of network states, given by a set $\Omega$, then $\pi(\omega)$ represents a probability
mass function and by stationarity we have:  
\[ \pi(\omega) = Pr[\omega(t) = \omega] \: \: \forall \omega \in \Omega, \forall t \in \{0, 1, 2, \ldots\} \]
In the case when $\Omega$ is possibly countably infinite, then we assume $\omega(t)$ is a 
random vector 
and $\pi(\omega)$ represents a probability density function.  The simplest model is when 
$\omega(t)$ is i.i.d. over slots, although the stationary assumption does not require independence
over slots. 

We further assume that the control decision $\alpha(t) \in \script{A}_{\omega(t)}$ can always be 
chosen to respect the backlog constraints, and that any algorithm that does not respect the 
backlog constraints can be modified to respect the backlog constraints without hindering performance.
This can be done simply by  never attempting to transmit more data than we
have, so that for all $k \in \{1, \ldots, K\}$ we have:
\begin{eqnarray}
\hat{b}_k(\alpha(t), \omega(t))  = b_k(t) = \tilde{b}_k(t) \label{eq:qc1} \\
\hat{y}_k(\alpha(t), \omega(t)) = y_k(t) = \tilde{y}_k(t) \label{eq:qc2} 
\end{eqnarray}
We define a policy $\alpha(t)$ that satisfies (\ref{eq:qc1})-(\ref{eq:qc2}) to be a \emph{policy
that respects queue backlog}.  It is clear that the queue dynamics (\ref{eq:multi-q-dynamics}) 
under such a policy become: 
\begin{equation} \label{eq:qc-dynamics} 
Q_k(t+1) = Q_k(t) - b_k(t) + y_k(t) + a_k(t) \: \: \forall k \in \{1, \ldots, K\}, \forall t
\end{equation}

\subsection{The Optimization Problem} 

Let $f(\bv{x})$ and $g_1(\bv{x}), g_2(\bv{x}), \ldots, g_L(\bv{x})$ be real-valued, 
continuous, and convex functions of $\bv{x} \in \mathbb{R}^M$ 
for some non-negative integer $L$ (if $L=0$ then there are no $g_l(\bv{x})$ functions).\footnote{These
functions might be defined over only a suitable subset of $\mathbb{R}^M$, such as the set
of all non-negative vectors.} 
Suppose we want to design a control algorithm that chooses $\alpha(t) \in \script{A}_{\omega(t)}$ 
over slots $t$ that solves the following general stochastic network optimization problem: 
\begin{eqnarray}
\hspace{-.3in}\mbox{Minimize:} &&\limsup_{t\rightarrow\infty} f(\overline{\bv{x}}(t)) \label{eq:opt1}  \\
\hspace{-.3in}\mbox{Subject to:} &1)&  \limsup_{t\rightarrow\infty} g_l(\overline{\bv{x}}(t)) \leq 0 \: \: \forall l \in \{1, \ldots, L\} \label{eq:opt2} \\
\hspace{-.3in}&2)& \alpha(t) \in \script{A}_{\omega(t)} \: \: \forall t \in \{0, 1, 2, \ldots\} \label{eq:opt2-point-5} \\
\hspace{-.3in}&3)& \mbox{All queues $Q_k(t)$ are mean rate stable} \label{eq:opt3} 
\end{eqnarray} 
where $\overline{\bv{x}}(t)$ is defined for $t>0$ by:
\[ \overline{\bv{x}}(t) \defequiv \frac{1}{t}\sum_{\tau=0}^{t-1} \expect{\bv{x}(\tau)} \]
We say that the problem is \emph{feasible} if there exists a control algorithm that satisfies
the constraints (\ref{eq:opt2})-(\ref{eq:opt3}).  
Assuming the problem is feasible, we define $f^{opt}$ as the infimum cost in (\ref{eq:opt1})
over all possible feasible policies that respect queue backlog. 

We define an \emph{$\omega$-only policy} as a policy that observes $\omega(t)$ and makes
a decision $\alpha(t) \in \script{A}_{\omega(t)}$ as a stationary and random function only of 
$\omega(t)$ (regardless of queue backlog, and hence not necessarily respecting the backlog
constraints (\ref{eq:qc1})-(\ref{eq:qc2})). 
By stationarity of $\omega(t)$, it follows that the expected values of $b_k(t)$, $y_k(t)$
$x_m(t)$ are the same on each slot under a particular $\omega$-only policy $\alpha^*(t) \in \script{A}_{\omega(t)}$: 
\begin{eqnarray*}
\overline{b}_k &=& \expect{\hat{b}_k(\alpha^*(t), \omega(t))} \\
\overline{y}_k &=& \expect{\hat{y}_k(\alpha^*(t), \omega(t))} \\
\overline{x}_m &=& \expect{\hat{x}_m(\alpha^*(t), \omega(t))} 
\end{eqnarray*}
where the expectation above is with respect to the stationary distribution $\pi(\omega)$
and the possibly randomized actions $\alpha^*(t)$. 
The next theorem characterizes all possible feasible algorithms (including algorithms
that are not $\omega$-only) in terms of $\omega$-only algorithms.

\begin{thm} \label{thm:omega-only} 
Suppose the problem (\ref{eq:opt1})-(\ref{eq:opt3}) is feasible with infimum cost $f^{opt}$, assumed
to be achievable arbitrarily closely by policies that respect the backlog 
constraints (\ref{eq:qc1})-(\ref{eq:qc2}). 
Then for all $\epsilon>0$ there exists
an $\omega$-only algorithm $\alpha^*(t)$ that satisfies the following for all $k \in \{1, \ldots, K\}$
and $l \in \{1, \ldots, L\}$: 
\begin{eqnarray}
g_l(\expect{\hat{\bv{x}}(\alpha^*(t), \omega(t))}) \leq \epsilon \label{eq:oo-1} \\
\lambda_k + \expect{\hat{y}_k(\alpha^*(t), \omega(t)) - \hat{b}_k(\alpha^*(t), \omega(t))} \leq \epsilon \label{eq:oo-2}  \\
f(\expect{\hat{\bv{x}}(\alpha^*(t), \omega(t))}) \leq f^{opt} + \epsilon \label{eq:oo-3} 
\end{eqnarray}
\end{thm}

Before proving Theorem \ref{thm:omega-only}, we present a preliminary lemma. 

\begin{lem} \label{lem:omega-only} For any algorithm that chooses $\alpha(\tau) \in \script{A}_{\omega(\tau)}$ over slots $\tau \in \{0, 1, 2, \ldots, \}$, and for any slot $t>0$, there exists an $\omega$-only policy 
$\alpha^*(t)$ that yields the following for all $k \in \{1, \ldots, K\}$, $m \in \{1, \ldots, M\}$: 
\begin{eqnarray*}
\frac{1}{t}\sum_{\tau=0}^{t-1} \expect{\hat{b}_k(\alpha(t), \omega(t))} &=& \expect{\hat{b}_k(\alpha^*(t), \omega(t))} \\
\frac{1}{t}\sum_{\tau=0}^{t-1} \expect{\hat{y}_k(\alpha(t), \omega(t))} &=& \expect{\hat{y}_k(\alpha^*(t), \omega(t))} \\
\frac{1}{t}\sum_{\tau=0}^{t-1} \expect{\hat{x}_m(\alpha(t), \omega(t))} &=& \expect{\hat{x}_m(\alpha^*(t), \omega(t))} 
\end{eqnarray*}
\end{lem} 
\begin{proof} (Lemma \ref{lem:omega-only}) For a given slot $t>0$, run the $\alpha(\tau)$ policy
and generate random quantities $[\tilde{\omega}, \tilde{\alpha}]$ as follows:  Uniformly pick
a time $T \in \{0, 1, \ldots, t-1\}$, and define $[\tilde{\omega}, \tilde{\alpha}] \defequiv [\omega(T), \alpha(T)]$, 
being the network state observed at the randomly chosen time $T$
and the corresponding 
network action $\alpha(T)$.  Because $\omega(\tau)$ is stationary, it follows that $\tilde{\omega}$ 
has the stationary distribution $\pi(\omega)$.  Now define the $\omega$-only policy $\alpha^*(t)$
to choose $\omega \in \script{A}_{\omega(t)}$ according to the conditional distribution of $\tilde{\alpha}$
given $\tilde{\omega}$ (generated from the joint distribution of $[\tilde{\omega}, \tilde{\alpha}]$). 
It follows that the expectations of $\hat{b}_k(\cdot)$, $\hat{y}_k(\cdot)$, $\hat{x}_m(\cdot)$ under
this $\omega$-only policy $\alpha^*(t)$ are as given in the statement of the lemma. 
\end{proof} 

\begin{proof} (Theorem \ref{thm:omega-only})
Fix $\epsilon>0$, and 
suppose $\alpha(t)$ is a policy that respects the queue backlog constraints (\ref{eq:qc1})-(\ref{eq:qc2}),  satisfies the feasibility constraints (\ref{eq:opt2})-(\ref{eq:opt3}), and such that: 
\[ \limsup_{t\rightarrow\infty} f(\overline{\bv{x}}(t)) \leq f^{opt} + \epsilon/2 \]
Then $\expect{Q_k(t)/t}\rightarrow 0$ for all $k \in \{1, \ldots, K\}$, and 
there is a slot $t^*>0$ such that: 
\begin{eqnarray}
\expect{Q_k(t^*)/t^*} \leq \epsilon & \forall k \in \{1, \ldots, K\} \label{eq:omega-only-foo} \\
g_l(\overline{\bv{x}}(t^*)) \leq \epsilon & \forall l \in \{1, \ldots, L\} \label{eq:omega-only-g}   \\
\frac{1}{t^*}\sum_{\tau=0}^{t^*-1}\expect{a_k(\tau)} \geq \lambda_k - \epsilon & \forall k \in \{1, \ldots, K\} \label{eq:omega-only-lambda}  \\
f(\overline{\bv{x}}(t)) \leq f_{opt} + \epsilon \label{eq:omega-only-f} 
\end{eqnarray}
where (\ref{eq:omega-only-lambda}) holds by (\ref{eq:lambda-k-assumption}). 
By (\ref{eq:qc-dynamics}) we have for all $k \in \{1, \ldots, K\}$: 
\begin{eqnarray*}
 \sum_{\tau=0}^{t^*-1} [a_k(\tau) + y_k(\tau) - b_k(\tau)] &=& Q_k(t^*) - Q_k(0) \\
 &\leq& Q_k(t^*) 
 \end{eqnarray*}
Taking expectations, dividing by $t^*$,  
and using (\ref{eq:omega-only-foo}) yields: 
\begin{eqnarray}
\frac{1}{t^*}\sum_{\tau=0}^{t^*-1}\expect{a_k(\tau) + y_k(\tau) - b_k(\tau)} \leq \epsilon  \label{eq:plug-t-star} 
\end{eqnarray} 
The above holds for all $k \in \{1, \ldots, K\}$.  Using Lemma \ref{lem:omega-only}, we know there
must be an $\omega$-only policy $\alpha^*(t)$ that satisfies for all $l \in \{1, \ldots, L\}$
and all $k \in \{1, \ldots, K\}$: 
\begin{eqnarray*}
g_l(\expect{\hat{\bv{x}}(\alpha^*(t), \omega(t))}) &\leq& \epsilon  \\
\lambda_k +  \expect{\hat{y}_k(\alpha^*(t), \omega(t)) - \hat{b}_k(\alpha^*(t), \omega(t))} &\leq& 2\epsilon \\
f(\overline{\bv{x}}(t)) &\leq& f^{opt} + \epsilon 
\end{eqnarray*}
Redefining $\epsilon' = 2\epsilon$ proves the result. 
\end{proof} 

We note that if $\omega(t)$ is defined by a periodic Markov chain, then it can be made stationary
by randomizing over the period.  In the case when $\omega(t)$ takes values in a finite set
$\Omega$, we can prove the same result without the stationary 
assumption \cite{neely-power-network-jsac}\cite{neely-thesis}. 

\section{The Capacity Region} 

Now define $\Lambda$ as the the set of all non-negative 
rate vectors $\bv{\lambda} = (\lambda_1, \ldots, \lambda_K)$ such that for all $\epsilon>0$, 
there exists an $\omega$-only
policy such that the constraints (\ref{eq:oo-1})-(\ref{eq:oo-2}) of Theorem \ref{thm:omega-only} 
hold.  It can be shown that this set $\Lambda$ is a closed set. 
By Theorem \ref{thm:omega-only}, we know that $\bv{\lambda} \in \Lambda$ is \emph{necessary} 
for the existence of an algorithm that makes
all queues $Q_k(t)$ mean rate stable and that satisfies the $g_l(\cdot)$ constraints 
(\ref{eq:opt2}).  Because, under some mild technical assumptions, mean 
rate stability is the \emph{weakest} form of stability, it follows that the constraint $\bv{\lambda} \in \Lambda$
is \emph{also} a necessary condition for stabilizing the network (subject to the $g_l(\cdot)$ constraints) 
under either rate stability, steady state stability, or strong stability. 

We now show that the set $\Lambda$ is the \emph{network capacity region}, in the sense that, under
some mild additional assumptions on the processes $a_k(t)$ and $\omega(t)$, it is possible to 
make all queues $Q_k(t)$ \emph{strongly stable} whenever $\bv{\lambda}$ is an \emph{interior point} 
of $\Lambda$. Because the technical assumptions we introduce will also imply that 
strong stability is the strongest stability definition, it follows that the same algorithm that makes
all queues $Q_k(t)$ strongly stable also makes them 
rate stable, mean rate stable, and steady state stable.  For simplicity of exposition, we treat the 
case when the functions $f(\bv{x})$, $g_l(\bv{x})$ are \emph{linear or affine} (the case of convex functions
is treated in \cite{now}\cite{neely-thesis}\cite{neely-fairness-infocom05}\cite{neely-universal-scheduling}). 
Note that $\expect{f(\bv{X})} = f(\expect{\bv{X}})$ for any linear or affine function and any random
vector $\bv{X}$. 

\subsection{The Decaying Memory Property} \label{section:decaying-memory}

Suppose that $\omega(t)$ has stationary distribution $\pi(\omega)$ as before, and that arrival
processes $a_k(t)$ have rates $\lambda_k$ that satisfy (\ref{eq:lambda-k-assumption}). 
Define $H(t)$ as the \emph{history} of the system over slots $\tau \in \{0, 1, \ldots, t-1\}$, consisting
of the initial queue states $Q_k(0)$ and all $\omega(\tau)$, $\alpha(\tau)$ values over this interval. 
We say that the processes $a_k(t)$ and $\omega(t)$ together with the functions
$\hat{b}_k(\cdot)$, $\hat{y}_k(\cdot)$, $g_l(\hat{\bv{x}}(\cdot))$,  have the \emph{decaying memory property} 
if for any $\omega$-only policy $\alpha^*(t)$ and  any $\delta>0$, there is an integer $T>0$ (which 
may depend on $\delta$ and $\alpha^*(t)$) such
that for all slots $t_0\geq 0$, all possible values of $H(t_0)$, and all $k \in \{1, \ldots, K\}$, 
$l \in \{1, \ldots, L\}$ we have: 
\begin{eqnarray}
\frac{1}{T}\sum_{\tau=t_0}^{t_0+T-1} \expect{a_k(\tau) + \hat{y}_k(\alpha^*(\tau), \omega(\tau))|H(t_0)} \nonumber \\
 - \frac{1}{T}\sum_{\tau=t_0}^{t_0+T-1} \expect{\hat{b}_k(\alpha^*(\tau), \omega(\tau))|H(t_0) } 
\leq \lambda_k  \nonumber \\
+ \mathbb{E}_{\pi}\left\{\hat{y}_k(\alpha^*(t), \omega(t)) - \hat{b}_k(\alpha^*(t), \omega(t))  \right\} + \delta \label{eq:dm1}  \\
\frac{1}{T}\sum_{\tau=t_0}^{t_0+T-1}\expect{g_l(\hat{\bv{x}}(\alpha(\tau), \omega(\tau)))|H(t_0)} \nonumber  \\
- \mathbb{E}_{\pi}\left\{ g_l(\hat{\bv{x}}(\alpha^*(t), \omega(t))) \right\}\leq \delta \label{eq:dm2} \\
\frac{1}{T}\sum_{\tau=t_0}^{t_0+T-1}\expect{f(\hat{\bv{x}}(\alpha(\tau), \omega(\tau)))|H(t_0)} \nonumber  \\
- \mathbb{E}_{\pi}\left\{f(\hat{\bv{x}}(\alpha^*(t), \omega(t))) \right\}\leq \delta \label{eq:dm3}
\end{eqnarray}
where $\mathbb{E}_{\pi}\left\{\cdot\right\}$ represents an expectation over the stationary distribution
$\pi(\omega)$ for $\omega(t)$.  Intuitively, the decaying memory property says that the affects of past
history decay over $T$ slots, so that all conditional time average expectations over this interval 
are within $\delta$ of their stationary values.  This property can be shown to hold when $\omega(t)$
and $a_k(t)$ are driven by a finite state irreducible (possibly not aperiodic) Markov chain, and when 
the conditional expectation of all processes is finite given the current state. 

\subsection{Second Moment Boundedness Assumptions} \label{section:second-moment}

We assume that for all $t$ and all (possibly randomized) control 
actions $\alpha(t) \in \script{A}_{\omega(t)}$, the second moment of the processes
are bounded, so that there is a finite constant $\sigma^2$ such that: 
\begin{eqnarray*}
\expect{\hat{y}_k(\alpha(t), \omega(t))^2} &\leq& \sigma^2 \\
\expect{\hat{b}_k(\alpha(t), \omega(t))^2} &\leq& \sigma^2 \\
\expect{g_l(\hat{\bv{x}}(\alpha(t), \omega(t)))^2} &\leq& \sigma^2 \\
\expect{a_k(t)^2} &\leq& \sigma^2
\end{eqnarray*}
Note that these second moment assumptions also ensure first moments are bounded. 
Finally, we assume the first moment of $f(\bv{x}(t))$ is bounded by finite constants $f_{min}$
and $f_{max}$, so that for all (possibly randomized) control actions $\alpha(t) \in \script{A}_{\omega(t)}$
we have: 
\begin{equation} \label{eq:f-bounded} 
 f_{min} \leq \expect{f(\hat{\bv{x}}(\alpha(t), \omega(t)))} \leq f_{max} 
 \end{equation}

\subsection{Lyapunov Drift} 

We use the framework of \cite{now}\cite{neely-thesis}\cite{neely-energy-it} 
to design a policy to solve the optimization 
problem (\ref{eq:opt1})-(\ref{eq:opt3}). To this end, for each inequality constraint (\ref{eq:opt2}) 
define a \emph{virtual queue} $Z_l(t)$ that is initially empty and that has update equation: 
\begin{equation} \label{eq:z-dynamics} 
Z_l(t+1) = \max[Z_l(t) + g_l(\bv{x}(t)), 0] 
\end{equation} 
where $\bv{x}(t) = \hat{\bv{x}}(\alpha(t), \omega(t))$. The actual queues $Q_k(t)$ are assumed
to satisfy (\ref{eq:multi-q-dynamics}). 

Define $\bv{\Theta}(t) \defequiv [\bv{Q}(t), \bv{Z}(t)]$ as a composite vector of all actual and 
virtual queues. Define a \emph{Lyapunov function} $L(\bv{\Theta}(t))$ as follows: 
\[ L(\bv{\Theta}(t)) \defequiv \frac{1}{2}\sum_{k=1}^KQ_k(t)^2 + \frac{1}{2}\sum_{l=1}^LZ_l(t)^2 \]
For a given integer $T>0$, define the \emph{$T$-step conditional Lyapunov drift} 
$\Delta_T(\bv{\Theta}(t))$ as follows:\footnote{Strictly
speaking, better notation is $\Delta_T(\bv{\Theta}(t), t)$, although we use the simpler notation
$\Delta_T(\bv{\Theta}(t))$ as a formal representation of the right hand side of (\ref{eq:drift-def}).} 
\begin{equation}\label{eq:drift-def} 
\Delta_T(\bv{\Theta}(t)) \defequiv \expect{L(\bv{\Theta}(t+T)) - L(\bv{\Theta}(t))|\bv{\Theta}(t)} 
\end{equation}

\begin{lem} \label{lem:driftyy} For any control policy and for any parameter $V\geq0$, 
$\Delta_T(\bv{\Theta}(t))$ satisfies: 
\begin{eqnarray}
&&  \Delta_T(\bv{\Theta}(t)) + V\sum_{\tau=t}^{t+T-1}\expect{f(\bv{x}(\tau))|\bv{\Theta}(t)}  \leq \nonumber \\
&& T^2\expect{\hat{B}|\bv{\Theta}(t)}  + V\sum_{\tau=t}^{t+T-1}\expect{f(\bv{x}(\tau))|\bv{\Theta}(t)} \nonumber \\
&& + \sum_{k=1}^K\sum_{\tau=t}^{t+T-1}\expect{Q_k(t)[a_k(\tau) + y_k(\tau) - b_k(\tau)]|\bv{\Theta}(t)}\nonumber \\
&& + \sum_{l=1}^L\sum_{\tau=t}^{t+T-1}\expect{Z_l(t)g_l(\bv{x}(\tau))|\bv{\Theta}(t)} \label{eq:drift} 
\end{eqnarray}
where $\hat{B}$ is a random variable that satisfies: 
\[ \expect{\hat{B}} \leq B \]
where $B$ is a finite constant that depends on the worst case  
second moment bounds of
the $a_k(\tau)$, $b_k(\tau)$, $y_k(\tau)$, $g_l(\bv{x}(\tau))$ processes, as described in 
more detail in the proof. 
\end{lem} 

\begin{proof} 
See Appendix G. 
\end{proof} 

The parameter $V>0$ will determine a performance-backlog 
tradeoff, as in \cite{neely-thesis}\cite{now}\cite{neely-energy-it}\cite{neely-fairness-infocom05}.

\subsection{The Drift-Plus-Penalty Algorithm} 

Consider now the following algorithm, defined in terms of given positive 
parameters $C>0$ and $V>0$. Every slot $t$, observe the current $\omega(t)$ and the 
current actual and virtual queues $Q_k(t)$, $Z_l(t)$, and choose a control action $\alpha(t) \in \script{A}_{\omega(t)}$ that comes within $C$ of minimizing the following expression: 
\begin{eqnarray*}
Vf(\hat{\bv{x}}(\alpha(t), \omega(t)))  + \sum_{l=1}^LZ_l(t)g_l(\hat{\bv{x}}(\alpha(t), \omega(t))    \\
+ \sum_{k=1}^KQ_k(t)[\hat{y}_k(\alpha(t), \omega(t)) - \hat{b}_k(\alpha(t), \omega(t))] 
\end{eqnarray*}
This algorithm is designed to come within an additive constant $C$ of minimizing the right 
hand side of (\ref{eq:drift}) over all actions $\alpha(t) \in \script{A}_{\omega(t)}$ that can be made, 
given the current queue states $\bv{\Theta}(t)$.  We call such a policy a \emph{$C$-approximation}. 
Note that a $0$-approximation is one that achieves the exact infimum on the right hand
side of (\ref{eq:drift}).  The notion of a $C$-approximation is introduced in case the infimum cannot
be achieved, or when the infimum is difficult to achieve exactly. 

 \begin{lem}  \label{lem:drift2} Suppose we use a  $C$-approximation every slot.  Then for any time $t$,  
 any integer $T>0$, and any $\bv{\Theta}(t)$, we have: 
 \begin{eqnarray}
 &&  \Delta_T(\bv{\Theta}(t)) + V\sum_{\tau=t}^{t+T-1}\expect{f(\bv{x}(\tau))|\bv{\Theta}(t)}  \leq \nonumber \\
&& CT + T^2\expect{\hat{B}|\bv{\Theta}(t)} \nonumber \\
&&   + T(T-1)\expect{\hat{D}|\bv{\Theta}(t)}  + V\sum_{\tau=t}^{t+T-1}\expect{f(\bv{x}^*(\tau))|\bv{\Theta}(t)} \nonumber \\
&& + \sum_{k=1}^KQ_k(t)\sum_{\tau=t}^{t+T-1}\expect{a_k(\tau) + y_k^*(\tau) - b_k^*(\tau)|\bv{\Theta}(t)}\nonumber \\
&& + \sum_{l=1}^LZ_l(t)\sum_{\tau=t}^{t+T-1}\expect{g_l(\bv{x}^*(\tau))|\bv{\Theta}(t)} \label{eq:drift2} 
 \end{eqnarray}
 where $\bv{x}^*(\tau)$, $y_k^*(\tau)$, $b_k^*(\tau)$ are values that correspond to any alternative
 policy for choosing $\alpha^*(\tau) \in \script{A}_{\omega(\tau)}$: 
 \begin{eqnarray*}
\bv{x}^*(\tau) &\defequiv&  \hat{\bv{x}}(\alpha^*(\tau), \omega(\tau)) \\
y_k^*(\tau) &\defequiv& \hat{y}_k(\alpha^*(\tau), \omega(\tau)) \\
b_k^*(\tau) &\defequiv& \hat{b}_k(\alpha^*(\tau), \omega(\tau)) 
 \end{eqnarray*}
 and where $\hat{D}$ is a random variable that satisfies: 
 \[ \expect{\hat{D}} \leq D \]
where $D$ is a finite constant related to the worst case 
 second moments of $a_k(t)$, $b_k(t)$, $y_k(t)$, 
 $g_l(\bv{x}(t))$, described in more detail in the proof. 
 \end{lem} 
 
 \begin{proof} See Appendix F. 
 \end{proof} 
 
 \subsection{Algorithm Performance} 
 
 In the following theorems, we assume the set $\Lambda$ is non-empty and has non-empty
 interior.  We say that a non-negative rate 
 vector $\bv{\lambda}$ is interior to $\Lambda$ if $\bv{\lambda} \in \Lambda$ and if there
 exists a value $d_{max}>0$ such that 
 $\bv{\lambda} + \bv{d}_{max} \in \Lambda$, where $\bv{d}_{max} = (d_{max}, d_{max}, \ldots, d_{max})$ 
 is a vector with all 
 entries equal to $d_{max}$. 
 We also assume
 the initial condition of the queues 
 satisfies $\expect{L(\bv{\Theta}(0))} < \infty$.\footnote{Note that $\expect{L(\bv{\Theta}(0))} = 0$
 if all queues are initially empty with probability 1.} 
 
 In the case when $L=0$, so that there are no $g_l(\cdot)$ constraints, we only require
 $\bv{\lambda}$ to be an interior point of $\Lambda$.  However, if $L>0$ we need a stronger
 assumption, related to a Slater-type condition of static optimization theory \cite{bertsekas-nonlinear}. 
 
 \emph{Assumption A1:} There is a constant $d_{max}>0$ and  
 an $\omega$-only policy that yields for all $l \in \{1, \ldots, L\}$, 
 $k \in \{1, \ldots, K\}$: 
\begin{eqnarray}
g_l(\expect{\hat{\bv{x}}(\alpha^*(t), \omega(t))}) \leq -d_{max}/2  \label{eq:a1-1} \\
\lambda_k + \expect{\hat{y}_k(\alpha^*(t), \omega(t)) - \hat{b}_k(\alpha^*(t), \omega(t))} \leq -d_{max}/2  
 \label{eq:a1-2} 
\end{eqnarray}

It is clear from Theorem \ref{thm:omega-only} that Assumption A1 holds whenever $\bv{\lambda}  + \bv{d}_{max} \in \Lambda$ and when $L=0$.   This is because if $\bv{\lambda} + \bv{d}_{max} \in \Lambda$, 
then for any $\epsilon>0$, Theorem \ref{thm:omega-only} implies the existence of an $\omega$-only
policy $\alpha^*(t)$ that satisfies for all $k \in \{1, \ldots, K\}$: 
\[ \lambda_k + d_{max} + \expect{\hat{y}_k(\alpha^*(t), \omega(t)) - \hat{b}_k(\alpha^*(t), \omega(t))} \leq \epsilon \]
and hence we can simply choose $\epsilon=d_{max}/2$ to yield (\ref{eq:a1-2}) (note that (\ref{eq:a1-1}) is 
irrelevant in the case $L=0$).

 \begin{thm} \label{thm:performance1} Suppose Assumption A1 holds. 
 Suppose we use a fixed parameter $V\geq 0$, and we implement a 
 $C$-approximation every slot $t$. 
 Suppose the system satisfies the decaying memory
 property of Section \ref{section:decaying-memory} and the boundedness assumptions
 of Section \ref{section:second-moment}. Then all actual and virtual queues are mean rate
 stable and so: 
 \begin{eqnarray}
 \limsup_{t\rightarrow\infty} g_l(\overline{\bv{x}}(t)) \leq 0 & \forall l \in \{1, \ldots, L\} \label{eq:achieve1} \\
 \lim_{t\rightarrow\infty} \expect{Q_k(t)/t} = 0 & \forall k \in \{1, \ldots, K\}  \label{eq:achieve2} 
 \end{eqnarray}
 Therefore, all desired constraints of problem (\ref{eq:opt1})-(\ref{eq:opt3}) are satisfied. 
 Further,  all queues are strongly stable and:
 \begin{eqnarray}
\limsup_{t\rightarrow\infty} \frac{1}{t}\sum_{\tau=0}^{t-1}\left[\sum_{k=1}^K\expect{Q_k(\tau)} + \sum_{l=1}^L\expect{Z_l(\tau)} \right] \leq \nonumber \\
\frac{[C + TB + (T-1)D + V(f_{max} - f_{min}) ]}{d_{max}/4} \label{eq:backlog-bound}  
\end{eqnarray}
 where $T$ is a positive integer related to $d_{max}$ and independent of $V$, and $B$ and $D$
 are constants defined in Lemmas \ref{lem:driftyy} and \ref{lem:drift2}.
 Further, for any $\epsilon$ such that $0 < \epsilon \leq d_{max}/4$, 
  there is a positive integer $T_{\epsilon}$, independent of $V$, 
 such that:  
 \begin{eqnarray}
 \limsup_{t\rightarrow\infty} f(\overline{\bv{x}}(\tau)) &\leq& f^{opt} + c_0\epsilon \nonumber \\
&& +  \frac{C + BT_{\epsilon}  + D(T_{\epsilon}-1)}{V}  \label{eq:f1} 
 \end{eqnarray}
 where $c_0$ is defined: 
 \[ c_0 \defequiv 4f_{max}/d_{max} + 1 \]
 The constants $T$, $T_{\epsilon}$ are related to the amount of time required for the memory to decay
 to a suitable proximity to a stationary distribution, as defined in the proof.
  \end{thm} 
 
 Theorem \ref{thm:performance1} shows that the algorithm makes all queues mean rate
 stable and satisfies all desired constraints for any $V\geq 0$ (including $V=0$).  The parameter
 $V$ is useful because the achieved 
 cost can be pushed to its optimal value $f^{opt}$ as $V\rightarrow \infty$, 
 as shown by (\ref{eq:f1}).
 Hence, 
 we can ensure the achieved cost is arbitrarily close to the optimum by choosing $V$ suitably
 large.  While a larger value of $V$ does not affect the constraints (\ref{eq:achieve1}), (\ref{eq:achieve2}), 
 it turns out that it creates a larger \emph{convergence time} over which time averages are close
 to meeting their constraints.  It also affects the average queue backlog sizes, so that the average
 backlog bound is $O(V)$, as shown in 
 (\ref{eq:backlog-bound}).   
 
In the i.i.d. case, it is known that the difference between the achieved cost and the optimal
cost $f^{opt}$ is $O(1/V)$, which establishes an 
$[O(1/V), O(V)]$ performance-backlog 
tradeoff \cite{neely-thesis}\cite{neely-fairness-infocom05}\cite{neely-energy-it}\cite{now}. 
For the non-i.i.d. case, for a given $\epsilon>0$, the bound (\ref{eq:f1}) shows that average
cost is within $O(1/V)$ of $f^{opt} + c_0\epsilon$.  However, the coefficient multiplier in the $O(1/V)$
term is linear in $T_{\epsilon}$, representing a ``mixing time'' required for time averages to be
within $O(\epsilon)$ of their stationary averages. If we choose $\epsilon = 1/V$, then this mixing
time itself can be a function of $V$ (we typically expect $T_{\epsilon} = O(\log(1/\epsilon)) = O(\log(V))$ 
if the network events are driven by a finite state Markov chain).  This would present an 
$[O(\log(V)/V), O(V)]$ performance-delay tradeoff.  Such a tradeoff is explicitly shown in 
\cite{rahul-cognitive-tmc}\cite{longbo-profit-allerton07} for particular types of networks, where
a worst case backlog bound of $O(V)$ is also derived.  Work in \cite{neely-mesh} treats 
a mobile network with non-ergodic mobility,  defines an ``instantaneous capacity region,''
and shows (using an analysis similar to i.i.d. analysis) that achieved utility is within $O(1/V)$
of the sum of optimal utilities associated with each instantaneous region.   

Our work \cite{neely-mesh}
also states an extension (without proof) 
that for ergodic mobility, the achieved utility is within $B_{mobile}/V$ of the optimum, 
where $B_{mobile}$ is a constant associated with the timescales of the mobility process, although
it does not compute this constant. 
This statement is in Theorem 1 part (f) of \cite{neely-mesh}.
The result (\ref{eq:f1}) above shows the 
constant $B_{mobile}$ can be defined for a given $\epsilon$ as: 
\[  B_{mobile} = c_0 \epsilon V + C + BT_{\epsilon} + D(T_{\epsilon}-1)\]
Defining $\epsilon \defequiv 1/V$ removes dependence on $V$ in the first term, but there is still
dependence on $V$ in the terms that are linear in $T_{\epsilon}$.  Typically, $T_{1/V}$ is  
$O(\log(V))$ for systems defined on finite state ergodic Markov chains.   
However, our statement in part (f) of \cite{neely-mesh} claims a constant $B_{mobile}$ can be found
that is \emph{independent of $V$} (yielding an $O(1/V)$ distance to optimality, rather than an 
$O(\log(V)/V)$ distance).  This is indeed the case, although it requires the $\omega(t)$ process
to be either i.i.d. over slots, or to be defined by an ergodic Markov chain 
with a finite state space.\footnote{We note that our statement in part (f) of \cite{neely-mesh} 
says that for ergodic mobility, a constant $B_{mobile}$ can be found that is independent of $V$. 
This should have been stated more precisely for mobility defined by a 
``finite state ergodic Markov chain.''  We regret this ambiguity. The full proof of this result 
is due to Longbo Huang and uses 
a variable frame drift analysis that is 
slightly different from the $T$-slot drift analysis suggested
in the proof outline in \cite{neely-mesh}.  The details are in preparation.}

 \begin{proof} (Theorem \ref{thm:performance1}) 
 By Assumption A1 there is an $\omega$-only
 policy $\alpha^*(t)$ such that for all $k \in \{1, \ldots, K\}$, $l \in \{1, \ldots, L\}$: 
\begin{eqnarray}
\expect{g_l(\hat{\bv{x}}(\alpha^*(t), \omega(t)))} \leq -d_{max}/2  \label{eq:oo-t1} \\
\lambda_k  + \expect{\hat{y}_k(\alpha^*(t), \omega(t)) - \hat{b}_k(\alpha^*(t), \omega(t))} \leq -d_{max}/2 \label{eq:oo-t2}
\end{eqnarray}
where we have used the fact that  $g_l(\cdot)$ is linear or affine to pass
the expectation through this function in  (\ref{eq:oo-t1}).

 Fix $\delta\defequiv d_{max}/4$, and choose a frame size $T>0$
 that satisfies the decaying memory properties 
 (\ref{eq:dm1})-(\ref{eq:dm3}) for this $\delta$ and this $\omega$-only policy $\alpha^*(t)$. 
 The value of $T$ depends on $\delta$, and so we could write $T_{\delta}$, although we 
 use $T$ below for notational simplicity (we note the dependence on $\delta$ again when needed). 
 By Lemma \ref{lem:drift2} we have: 
 \begin{eqnarray}
 &&  \Delta_T(\bv{\Theta}(t)) + V\sum_{\tau=t}^{t+T-1}\expect{f(\bv{x}(\tau))|\bv{\Theta}(t)}  \leq CT \nonumber \\
&& + T^2\expect{\hat{B}|\bv{\Theta}(t)}   + T(T-1)\expect{\hat{D}|\bv{\Theta}(t)} \nonumber \\
&&   + V\sum_{\tau=t}^{t+T-1}\expect{f(\bv{x}^*(\tau))|\bv{\Theta}(t)} \nonumber \\
&& + \sum_{k=1}^KQ_k(t)\sum_{\tau=t}^{t+T-1}\expect{a_k(\tau) + y_k^*(\tau) - b_k^*(\tau)|\bv{\Theta}(t)}\nonumber \\
&& + \sum_{l=1}^LZ_l(t)\sum_{\tau=t}^{t+T-1}\expect{g_l(\bv{x}^*(\tau))|\bv{\Theta}(t)} \label{eq:drift3} 
 \end{eqnarray}    
 By the decaying memory property, the conditional expectations given $\bv{\Theta}(t)$ are within
 $\delta = d_{max}/4$ of their stationary averages, and so (by using (\ref{eq:oo-t1})-(\ref{eq:oo-t2})): 
 \begin{eqnarray}
 &&  \Delta_T(\bv{\Theta}(t)) + V\sum_{\tau=t}^{t+T-1}\expect{f(\bv{x}(\tau))|\bv{\Theta}(t)}  \leq CT \nonumber \\
&& + T^2\expect{\hat{B}|\bv{\Theta}(t)}  + T(T-1)\expect{\hat{D}|\bv{\Theta}(t)}    + VTf_{max}  \nonumber \\
&& - \sum_{k=1}^KQ_k(t)Td_{max}/4- \sum_{l=1}^LZ_l(t)Td_{max}/4 \label{eq:drift4} 
 \end{eqnarray}
 Rearranging terms  yields: 
 \begin{eqnarray*}
 \Delta_T(\bv{\Theta}(t))  +\frac{Td_{max}}{4}\left[\sum_{k=1}^KQ_k(t)+ \sum_{l=1}^LZ_l(t)\right] \leq \\
 CT + T^2\expect{\hat{B}|\bv{\Theta}(t)} + T(T-1)\expect{\hat{D}|\bv{\Theta}(t)}  \\
 + VT(f_{max}- f_{min}) 
 \end{eqnarray*}
 Taking expectations of both sides and using the law of iterated expectations yields: 
 \begin{eqnarray*}
 \expect{L(\bv{\Theta}(t+T))} - \expect{L(\bv{\Theta}(t))} \\
 + \frac{Td_{max}}{4}\expect{\sum_{k=1}^KQ_k(t) + \sum_{l=1}^LZ_l(t)} \\
 \leq  CT + T^2B  +T(T-1)D   + VT(f_{max} - f_{min}) 
 \end{eqnarray*}
 The above holds for all slots $t \geq 0$.  Define $t_i = t_0  + iT$ for $i \in \{0, 1, 2, \ldots\}$, 
 where $t_0\in \{0, 1, \ldots, T-1\}$. 
 We thus have: 
 \begin{eqnarray*}
 \expect{L(\bv{\Theta}(t_{i+1}))} - \expect{L(\bv{\Theta}(t_i))} \\
 + \frac{Td_{max}}{4}\left[\sum_{k=1}^K\expect{Q_k(t_i)} + \sum_{l=1}^L\expect{Z_l(t_i)}\right] \\
 \leq  CT + T^2B + T(T-1)D + VT(f_{max} - f_{min}) 
 \end{eqnarray*}
 Summing over $i \in \{0,  \ldots, J-1\}$ yields: 
 \begin{eqnarray*}
 \expect{L(\bv{\Theta}(t_{J}))} - \expect{L(\bv{\Theta}(t_0))} \\
 + \frac{Td_{max}}{4}\sum_{i=0}^{J-1}\left[\sum_{k=1}^K\expect{Q_k(t_i)} + \sum_{l=1}^L\expect{Z_l(t_i)}\right] \\
 \leq  J[CT + T^2B + T(T-1)D+ VT(f_{max} - f_{min}) ]
 \end{eqnarray*}
 Rearranging terms and using the fact that $\expect{L(\bv{\Theta}(t_J))} \geq 0$ yields: 
  \begin{eqnarray*}
\sum_{i=0}^{J-1}\left[\sum_{k=1}^K\expect{Q_k(t_i)} + \sum_{l=1}^L\expect{Z_l(t_i)}\right] \\
 \leq  \frac{J[CT + BT^2 + T(T-1)D + VT(f_{max} - f_{min}) ]}{Td_{max}/4} \\
 + \frac{\expect{L(\bv{\Theta}(t_0))}}{Td_{max}/4}
 \end{eqnarray*}
 Summing over all $t_0\in\{0, 1, \ldots, T-1\}$ and dividing by $JT$ yields: 
  \begin{eqnarray*}
\frac{1}{JT}\sum_{\tau=0}^{JT-1}\left[\sum_{k=1}^K\expect{Q_k(\tau)} + \sum_{l=1}^L\expect{Z_l(\tau)}\right] \\
 \leq  \frac{[C + TB + (T-1)D + V(f_{max} - f_{min}) ]}{d_{max}/4} \\
 + \frac{1}{JT}\sum_{t_0=0}^{T-1}\frac{\expect{L(\bv{\Theta}(t_0))}}{Td_{max}/4}
 \end{eqnarray*}
The above holds for all positive integers $J$.  Taking a $\limsup$ as $J\rightarrow\infty$ yields:\footnote{While the above only samples the time average expectation over slots that are multiples
of $T$, it is easy to see that the $\limsup$ time average over all slots is the same, as the queue cannot
change by much over the course of $T$ slots.} 
\begin{eqnarray*}
\limsup_{t\rightarrow\infty} \frac{1}{t}\sum_{\tau=0}^{t-1}\left[\sum_{k=1}^K\expect{Q_k(\tau)} + \sum_{l=1}^L\expect{Z_l(\tau)} \right] \leq \\
\frac{[C + TB  + (T-1)D + V(f_{max} - f_{min}) ]}{d_{max}/4} \\
\end{eqnarray*}
This proves  (\ref{eq:backlog-bound}), and hence 
proves strong stability of all queues.  By Theorem \ref{thm:strong-stability}, since the second
moments of the arrival and service processes for all queues are bounded, we know mean rate stability 
also holds for all queues.  Because queues $Z_l(t)$ are mean rate stable, and these queues have
update equation (\ref{eq:z-dynamics}), we know from Theorem \ref{thm:gen-nec-rate} that: 
\[ \limsup_{t\rightarrow\infty} \frac{1}{t}\sum_{\tau=0}^{t-1} \expect{g_l(\bv{x}(\tau))} \leq 0 \]
Passing the expectation through the linear function $g_l(\bv{x})$ proves (\ref{eq:achieve1}). 

It remains only to prove (\ref{eq:f1}).  Fix $\epsilon>0$, and assume that 
$\epsilon< d_{max}/4$.  Note that Theorem \ref{thm:omega-only}
implies the existence of an $\omega$-only algorithm $\alpha'(t)$ that satisfies: 
\begin{eqnarray}
\expect{g_l(\hat{\bv{x}}(\alpha'(t), \omega(t)))} \leq \epsilon  \label{eq:a2-1} \\
\lambda_k + \expect{\hat{y}_k(\alpha'(t), \omega(t)) - \hat{b}_k(\alpha^*(t), \omega(t))} \leq \epsilon  
 \label{eq:a2-2} \\
 \expect{f(\hat{\bv{x}}(\alpha'(t), \omega(t)))} \leq f^{opt} + \epsilon \label{eq:a2-3} 
\end{eqnarray}
where we have used the fact that $f(\cdot)$, $g_l(\cdot)$ are linear or affine to pass expectations
through them.  Now define an $\omega$-only policy $\alpha^{\star}(t)$ as follows: 
\[ \alpha^{\star}(t)  = \left\{ \begin{array}{ll}
                          \alpha^*(t) &\mbox{ with probability $\theta$} \\
                             \alpha'(t)  & \mbox{ with probability $1-\theta$} 
                            \end{array}
                                 \right.\]
where $\alpha^*(t)$ is the algorithm of Assumption A1, and where 
 $\theta$ is defined: 
\begin{equation} \label{eq:theta}  
\theta \defequiv 4\epsilon/d_{max} 
\end{equation} 
Note that $\theta$ is a valid probability because $0 < \epsilon \leq d_{max}/4$. 
Therefore, under policy $\alpha^{\star}(t)$ we have (combining (\ref{eq:a2-1})-(\ref{eq:a2-3}) 
and (\ref{eq:a1-1})-(\ref{eq:a1-2})): 
\begin{eqnarray}
\expect{g_l(\hat{\bv{x}}(\alpha^{\star}(t), \omega(t)))} \leq \nonumber \\
 (1-\theta)\epsilon -\theta d_{max}/2 \leq -\epsilon \label{eq:a3-1} \\
\lambda_k + \expect{\hat{y}_k(\alpha^{\star}(t), \omega(t)) - \hat{b}_k(\alpha^*(t), \omega(t))} \nonumber \\
 \leq (1-\theta)\epsilon -\theta d_{max}/2  \leq -\epsilon  \label{eq:a3-2} \\
 \expect{f(\hat{\bv{x}}(\alpha^{\star}(t), \omega(t)))} \leq (1-\theta)f^{opt}  + \theta f_{max}  \label{eq:a3-3} 
\end{eqnarray}
where we have used the fact that $\theta = 4\epsilon/d_{max}$ to conclude that: 
\[ \theta d_{max}/2 = 2\epsilon \]

Now fix $\delta= \epsilon$, and define $T_{\epsilon}$ as the 
value that satisfies the decaying memory properties (\ref{eq:dm1})-(\ref{eq:dm3}) for this $\delta$
and for the $\omega$-only policy $\alpha^{\star}(t)$. 
By Lemma \ref{lem:drift2}: 
 \begin{eqnarray}
 &&  \Delta_{T_{\epsilon}}(\bv{\Theta}(t)) + V\sum_{\tau=t}^{t+T_{\epsilon}-1}\expect{f(\bv{x}(\tau))|\bv{\Theta}(t)}  \leq CT  \nonumber \\
&& + T_{\epsilon}^2\expect{\hat{B}|\bv{\Theta}(t)} + 
T_{\epsilon}(T_{\epsilon}-1)\expect{\hat{D}|\bv{\Theta}(t)} \nonumber \\  
&& + V\sum_{\tau=t}^{t+T_{\epsilon}-1}\expect{f(\bv{x}^{\star}(\tau))|\bv{\Theta}(t)} \nonumber \\
&& + \sum_{k=1}^KQ_k(t)\sum_{\tau=t}^{t+T_{\epsilon}-1}\expect{a_k(\tau) + y_k^{\star}(\tau) - b_k^{\star}(\tau)|\bv{\Theta}(t)}\nonumber \\
&& + \sum_{l=1}^LZ_l(t)\sum_{\tau=t}^{t+T_{\epsilon}-1}\expect{g_l(\bv{x}^{\star}(\tau))|\bv{\Theta}(t)} \label{eq:drift5} 
 \end{eqnarray}
 Noting that the decaying memory property ensures the above conditional expectations
 are within $\delta = \epsilon$ of their stationary averages (\ref{eq:a3-1})-(\ref{eq:a3-3}), we have: 
 \begin{eqnarray*}
 &&  \Delta_{T_{\epsilon}}(\bv{\Theta}(t)) + V\sum_{\tau=t}^{t+T_{\epsilon}-1}\expect{f(\bv{x}(\tau))|\bv{\Theta}(t)}  \leq \nonumber \\
&& CT + T_{\epsilon}^2\expect{\hat{B}|\bv{\Theta}(t)}  + T_{\epsilon}(T_{\epsilon}-1)\expect{\hat{D}|\bv{\Theta}(t)} \\
&& + VT_{\epsilon}(f^{opt} + \theta f_{max} + \epsilon) 
 \end{eqnarray*}
Taking expectations gives: 
\begin{eqnarray*}
 && \expect{L(\bv{\Theta}(t+T_{\epsilon}))} - \expect{L(\bv{\Theta}(t))} \nonumber \\
 && + V\sum_{\tau=t}^{t+T_{\epsilon}-1}\expect{f(\bv{x}(\tau))}  \leq CT + T_{\epsilon}^2B+T_{\epsilon}(T_{\epsilon}-1)D \nonumber \\
&&    + VT_{\epsilon}(f^{opt} + \theta f_{max} + \epsilon)
 \end{eqnarray*}
 As before, we substitute $t=t_i$ for $i \in \{0, 1, 2, \ldots, \}$ for some value $t_0 \in \{0, 1, \ldots T-1\}$, 
 and sum over $i \in \{0, 1, \ldots, J-1\}$ and $t_0 \in \{0, 1, \ldots, T-1\}$ to get: 
 \begin{eqnarray*}
\frac{V}{JT_{\epsilon}}\sum_{\tau=0}^{JT_{\epsilon}-1} \expect{f(\bv{x}(\tau))} \leq [C + T_{\epsilon}B+(T_{\epsilon}-1)D] \\
+ V(f^{opt} + \theta f_{max} + \epsilon) + \frac{1}{JT_{\epsilon}}\sum_{t_0=0}^{T_{\epsilon}-1} \expect{L(\bv{\Theta}(t_0))}
 \end{eqnarray*}
 Dividing by $V$ and taking a limit as $J\rightarrow\infty$ yields: 
 \begin{eqnarray*}
 \limsup_{t\rightarrow\infty} f(\overline{\bv{x}}(\tau)) \leq f^{opt} + \theta f_{max} + \epsilon  \\
 +  \frac{C+ T_{\epsilon}B+(T_{\epsilon}-1)D}{V} 
 \end{eqnarray*}
 where we have used the fact that $f(\bv{x})$ is linear or affine to pass the 
 time average expectation through it.  Using $\theta = 4\epsilon/d_{max}$ proves (\ref{eq:f1}).  
 \end{proof}

 \section{Exercises} \label{section:exercise}

 \begin{exer} \label{ex:inequality-comparison} (Inequality comparison)  Let $Q(t)$ satisfy (\ref{eq:q-dynamics}) with
 server process $b(t)$ and arrival process $a(t)$. 
 Let $\tilde{Q}(t)$ be another queueing system with the same server process $b(t)$ but with an arrival process $\tilde{a}(t) = a(t)  + z(t)$, 
 where $z(t) \geq 0$ for all $t \in \{0, 1, 2, \ldots\}$.  Assuming that $Q(0) = \tilde{Q}(0)$, prove that $Q(t) \leq \tilde{Q}(t)$ 
for all $t \in \{0, 1, 2, \ldots\}$.  
 \end{exer} 
 
 \begin{exer} \label{ex:rate-stable} (Proving sufficiency for Theorem \ref{thm:rate-stability}a) 
 Let $Q(t)$ satisfy (\ref{eq:q-dynamics})
 with arrival and server processes with well defined time averages $a_{av}$ and $b_{av}$. Suppose that $a_{av} \leq b_{av}$. 
 Fix $\epsilon>0$, and define $Q_{\epsilon}(t)$ as a queue with $Q_{\epsilon}(0) = Q(0)$, and 
 with the same server process $b(t)$ but with an arrival process
 $\tilde{a}(t) = a(t) + (b_{av} - a_{av}) + \epsilon$ for all $t$. 
  
 a) Compute the time average of $\tilde{a}(t)$.
 
 b) Assuming the result of Theorem \ref{thm:rate-stability}b, compute $\lim_{t\rightarrow\infty} Q_{\epsilon}(t)/t$. 
 
 c) Use the result of part (b) and Exercise \ref{ex:inequality-comparison} to prove that $Q(t)$ is rate stable.  
 \end{exer}

 \begin{exer} \label{ex:rate-stability-b}  (Proof of Theorem \ref{thm:rate-stability}b) 
 Let $Q(t)$ be a queue that satisfies (\ref{eq:q-dynamics}).  Assume time averages of 
 $a(t)$ and $b(t)$ are given by finite constants $a_{av}$ and $b_{av}$, respectively.   
 
 a) Use the following equation to prove that $\lim_{t\rightarrow \infty} a(t)/t = 0$ with probability 1:
 \[ \frac{1}{t+1} \sum_{\tau=0}^{t} a(\tau) = \left(\frac{t}{t+1}\right)\frac{1}{t}\sum_{\tau=0}^{t-1} a(\tau) + \left(\frac{t}{t+1}\right)\frac{a(t)}{t} \]
 
 b) Suppose that $\tilde{b}(t_i) < b(t_i)$ for some slot $t_i$.  Use (\ref{eq:q-dynamics}) to compute $Q(t_i+1)$.  
 
 c) Use part (b) to show that if $\tilde{b}(t_i) < b(t_i)$, then: 
 \[ a(t_i) \geq Q(0) + \sum_{\tau=0}^{t_i} [a(\tau) - b(\tau)]   \]
 Conclude that if $\tilde{b}(t_i) < b(t_i)$ for an infinite number of slots $t_i$, then $a_{av} \leq b_{av}$. 
 
 d) Use part (c) to conclude that if $a_{av} > b_{av}$, there is some slot $t^*\geq 0$ such that
 for all $t \geq t^*$ we have: 
 \[ Q(t) = Q(t^*) + \sum_{\tau=t^*}^{t-1} [a(\tau) - b(\tau)] \]
 Use this to prove the result of Theorem \ref{thm:rate-stability}b.
 \end{exer} 
 
 \begin{exer} \label{ex:strong-stability-implies-steady-state} (Strong stability implies steady state stability) Prove
 that strong stability implies steady state stability using the fact that  
 $\expect{Q(\tau)} \geq MPr[Q(\tau)>M]$. 
 \end{exer} 
 
\section*{Appendix A --- Proof of Theorem \ref{thm:rs-implies-mrs}}

Here we prove Theorem \ref{thm:rs-implies-mrs}.  Note that rate stability implies
that for any $\delta>0$, we have:\footnote{In fact, the result of parts (a) and (b) 
of Theorem \ref{thm:rs-implies-mrs} hold equally if the assumption that $Q(t)$ 
is rate stable is replaced by the weaker assumption (\ref{eq:rs-weaker}).}
\begin{eqnarray} 
\lim_{t\rightarrow\infty}Pr[Q(t)/t>\delta] = 0 \label{eq:rs-weaker}
\end{eqnarray}
For a given $\delta>0$, define the event $\script{E}_t \defequiv \{Q(t)/t > \delta\}$, 
so that $\lim_{t\rightarrow\infty} Pr[\script{E}_t]=0$.  Define $\script{E}_t^c \defequiv \{Q(t)/t \leq \delta\}$. 
\begin{lem} \label{lem:prelim-limit} 
If $Q(t)$ is non-negative and satisfies (\ref{eq:rs-weaker}), and if:
\begin{equation} \label{eq:tail-cond} 
 \lim_{t\rightarrow\infty} \expect{Q(t)/t|\script{E}_t}Pr[\script{E}_t] = 0 
 \end{equation} 
then $Q(t)$ is mean rate stable. 
\end{lem} 
\begin{proof} 
We have for a given $\delta>0$: 
\begin{eqnarray*}
0 \leq \expect{Q(t)/t} &=& \expect{Q(t)/t|\script{E}_t^c}Pr[\script{E}_t^c] \\
&& + \expect{Q(t)/t|\script{E}_t}Pr[\script{E}_t] \\
&\leq& \delta + \expect{Q(t)/t|\script{E}_t}Pr[\script{E}_t]
\end{eqnarray*}
Taking a $\limsup$ of both sides and using (\ref{eq:tail-cond}) yields: 
\[ 0 \leq \limsup_{t\rightarrow\infty} \expect{Q(t)/t} \leq \delta \]
This holds for all $\delta>0$.  Thus, $\lim_{t\rightarrow\infty} \expect{Q(t)/t} = 0$, proving
mean rate stability. 
\end{proof} 

Thus, to prove Theorem \ref{thm:rs-implies-mrs}, it suffices to prove that (\ref{eq:tail-cond}) 
holds under the assumptions of parts (a) and (b) of the theorem.   To this end, we have a 
preliminary lemma.

\begin{lem} \label{lem:X-RV}  If $X$ is a non-negative 
random variable such that $\expect{X^{1+\epsilon}}  < \infty$ for some value $\epsilon>0$, 
then for any event $\script{E}$ with a well defined probability $Pr[\script{E}]$, we have: 
\[ \expect{X|\script{E}}Pr[\script{E}] \leq \expect{X^{1+\epsilon}}^{1/(1+\epsilon)}Pr[\script{E}]^{\epsilon/(1+\epsilon)} \] 
\end{lem} 
\begin{proof} 
If $Pr[\script{E}]=0$, then the result is obvious.  Suppose now that $Pr[\script{E}]>0$. 
We have: 
\begin{eqnarray*}
 \expect{X^{1+\epsilon}} &=& \expect{X^{1+\epsilon}|\script{E}}Pr[\script{E}] + \expect{X^{1+\epsilon}|\script{E}^c}Pr[\script{E}^c] \\
 &\geq& \expect{X^{1+\epsilon}|\script{E}}Pr[\script{E}] 
\end{eqnarray*}
Therefore: 
\[ \expect{X^{1+\epsilon}|\script{E}} \leq \frac{\expect{X^{1+\epsilon}}}{Pr[\script{E}]} \]
However, by Jensen's inequality for the convex function $f(x) = x^{1+\epsilon}$
for $x \geq 0$, we have: 
\[ \expect{X|\script{E}}^{1+\epsilon} \leq \expect{X^{1+\epsilon}|\script{E}} \]
Thus: 
\[ \expect{X|\script{E}}^{1+\epsilon} \leq \frac{\expect{X^{1+\epsilon}}}{Pr[\script{E}]} \]
Hence: 
\[ \expect{X|\script{E}} \leq \left(\frac{\expect{X^{1+\epsilon}}}{Pr[\script{E}]}\right)^{1/(1+\epsilon)}\]
Multiplying both sides by $Pr[\script{E}]$ proves the result. 
\end{proof}

We now prove part (a) of Theorem \ref{thm:rs-implies-mrs}.

\begin{proof} (Theorem \ref{thm:rs-implies-mrs}(a)) 
For simplicity assume that $Q(0) = 0$. 
Note that: 
\begin{equation*}
\frac{Q(t)}{t} \leq \frac{1}{t}\sum_{\tau=0}^{t-1} [a(\tau) + b^-(\tau)] 
\end{equation*} 
Define $X(\tau) \defequiv a(\tau) + b^-(\tau)$.  Thus: 
\begin{equation} \label{eq:rs-to-mrs} 
\frac{Q(t)}{t} \leq \frac{1}{t}\sum_{\tau=0}^{t-1}X(\tau) 
\end{equation} 
Now suppose there 
are constants $\epsilon>0$, $C>0$ such that: 
\begin{equation} \label{eq:c-bound} 
 \expect{X(\tau)^{1+\epsilon}} \leq C \: \: \mbox{ for all $\tau$} 
 \end{equation} 
Fix $\delta>0$ and define the event $\script{E}_t \defequiv \{Q(t)/t >\delta\}$. Thus:
\begin{eqnarray}
 \expect{\frac{Q(t)}{t}|\script{E}_t}Pr[\script{E}_t] &\leq& \frac{1}{t}\sum_{\tau=0}^{t-1}\expect{X(\tau)|\script{E}_t}Pr[\script{E}_t] \label{eq:reuse-rs} \\
&\leq& \frac{1}{t}\sum_{\tau=0}^{t-1} C^{1/(1+\epsilon)} Pr[\script{E}_t]^{\epsilon/(1+\epsilon)}  \nonumber \\
&=& C^{1/(1+\epsilon)}Pr[\script{E}_t]^{\epsilon/(1+\epsilon)}  \nonumber
\end{eqnarray}
where the first inequality follows by (\ref{eq:rs-to-mrs}) and the second inequality 
uses Lemma \ref{lem:X-RV} together with (\ref{eq:c-bound}). 
Taking a limit of the above and using the fact that $Pr[\script{E}_t]\rightarrow 0$ yields: 
\[ \lim_{t\rightarrow\infty} \expect{Q(t)/t|\script{E}_t}Pr[\script{E}_t] = 0 \]
and therefore $Q(t)$ is mean rate stable by Lemma \ref{lem:prelim-limit}.
\end{proof} 

To prove part (b) of Theorem \ref{thm:rs-implies-mrs}, we need another preliminary 
lemma.

\begin{lem} \label{lem:prelim2} If $X$ is a non-negative random variable and $\script{E}$
is any event with a well defined probability $Pr[\script{E}]$, then for any $x > 0$ we have: 
\[ \expect{X|\script{E}}Pr[\script{E}] \leq \expect{X|X>x}Pr[X>x] + xPr[\script{E}] \]
\end{lem} 
\begin{proof} 
Fix a value $x>0$. 
Define indicator functions $1_{\script{E}}$ and $1_{\{X > x\}}$ as follows: 
\begin{eqnarray*}
 1_{\script{E}} &\defequiv&  \left\{ \begin{array}{ll}
                          1&\mbox{ if event $\script{E}$ is true} \\
                             0  & \mbox{ otherwise} 
                            \end{array}
                                 \right.  \\
  1_{\{X>x\}} &\defequiv& \left\{ \begin{array}{ll}
                          1 &\mbox{ if $X > x$} \\
                             0  & \mbox{ otherwise} 
                            \end{array}
                                 \right.  
  \end{eqnarray*}
  Define $1_{\{X\leq x\}} \defequiv 1 - 1_{\{X>x\}}$. 
  Then: 
  \begin{eqnarray*}
   X1_{\script{E}} &=& X1_{\script{E}}(1_{\{X> x\}} + 1_{\{X\leq x\}}) \\
   &=& X1_{\script{E}}1_{\{X> x\}} + X1_{\script{E}}1_{\{X\leq x\}} \\
   &\leq& X1_{\{X >  x\}} + x1_{\script{E}}
\end{eqnarray*}
Therefore: 
\[ \expect{X 1_{\script{E}}} \leq \expect{X 1_{\{X>x\}}} + xPr[\script{E}]  \]
Thus: 
\[ \expect{X|\script{E}}Pr[\script{E}] \leq \expect{X|X>x}Pr[X>x] + xPr[\script{E}] \]
\end{proof} 

\begin{proof} (Theorem \ref{thm:rs-implies-mrs}(b)) 
Fix $\delta>0$ and define the event $\script{E}_t \defequiv \{Q(t)/t > \delta\}$. 
From (\ref{eq:reuse-rs}) we have: 
\begin{eqnarray} 
 \expect{\frac{Q(t)}{t}|\script{E}_t}Pr[\script{E}_t] &\leq& \frac{1}{t}\sum_{\tau=0}^{t-1}\expect{X(\tau)|\script{E}_t}Pr[\script{E}_t] \label{eq:thm1b} 
\end{eqnarray}
where we recall that $X(\tau) \defequiv a(\tau) + b^-(\tau)$. 
Now fix any (arbitrarily large) $x>0$.  From Lemma \ref{lem:prelim2} we have: 
\begin{eqnarray*}
 && \hspace{-.3in} \expect{X(\tau)|\script{E}_t} Pr[\script{E}_t] \\
 &\leq& \expect{X(\tau)|X>x}Pr[X(\tau)>x] + xPr[\script{E}_t] \\
 &\leq& \expect{Y|Y>x}Pr[Y>x] + xPr[\script{E}_t]
\end{eqnarray*}
where the final equality has used the assumption about the 
random variable $Y$ in 
Theorem \ref{thm:rs-implies-mrs} part (b). 
Plugging this into (\ref{eq:thm1b})
yields: 
\begin{eqnarray*}
\expect{\frac{Q(t)}{t}|\script{E}_t}Pr[\script{E}_t] \leq \expect{Y|Y>x}Pr[Y>x] \\
+  \frac{x}{t}\sum_{\tau=0}^{t-1}Pr[\script{E}_t] 
\end{eqnarray*}
Because $Pr[\script{E}_t]\rightarrow 0$ as $t\rightarrow \infty$, its time average
also converges to $0$.  Thus, taking a $\limsup$ of both sides of the 
above inequality as $t\rightarrow\infty$ yields: 
\begin{eqnarray*}
\limsup_{t\rightarrow\infty} \expect{\frac{Q(t)}{t}|\script{E}_t} Pr[\script{E}_t] \leq \expect{Y|Y>x}Pr[Y>x] 
\end{eqnarray*}
The above holds for all $x>0$. The 
fact that $\expect{Y}<\infty$ ensures that the right hand side of the above inequality vanishes
as $x\rightarrow\infty$.  Taking a limit as $x\rightarrow\infty$ thus proves: 
\[ \limsup_{t\rightarrow\infty} \expect{\frac{Q(t)}{t}|\script{E}_t}Pr[\script{E}_t] = 0 \]
and hence: 
\[ \lim_{t\rightarrow\infty} \expect{\frac{Q(t)}{t}|\script{E}_t}Pr[\script{E}_t] = 0 \]
This together with Lemma \ref{lem:prelim-limit} proves that $Q(t)$ is mean rate stable.
\end{proof}

\section*{Appendix B --- Proof of Theorem \ref{thm:strong-stability}(c)} 

Here we prove part (c) of Theorem \ref{thm:strong-stability}. 
The proof is similar to our previous proof 
in \cite{neely-downlink-ton}.

\begin{proof} (Theorem \ref{thm:strong-stability} part (c)) 
Suppose there is a finite constant $C>0$ such that $\expect{b(t) - a(t)} \leq C$ for all $t$,
and that $Q(t)$ is \emph{not} mean rate stable. It follows that there is an $\epsilon>0$ such
that $\expect{Q(t_k)/t_k} \geq \epsilon$ for an infinite collection of times $t_k$. 
For any $t_k$ and any $t> t_k$ we have by (\ref{eq:io-b2}): 
\[ Q(t) \geq Q(t_k) - \sum_{\tau=t_k}^{t-1} [b(\tau) - a(\tau)] \]
Thus, for any $t \geq t_k$ we have: 
\[ \expect{Q(t)} \geq \epsilon t_k - (t-t_k)C \]
Now fix any (arbitrarily large) value $M>0$. Then for sufficiently large
$k$ we have $\epsilon t_k > M$, and 
$\epsilon t_k - (t-t_k)C \geq M$ whenever:
\begin{equation} \label{eq:t-cond}
 t_k \leq t \leq \frac{(C+\epsilon)t_k - M}{C} = (1+\epsilon/C)t_k - M/C 
\end{equation} 
Hence, $\expect{Q(t)} \geq M$ whenever (\ref{eq:t-cond}) holds. 
Define $\hat{t}_k$ as: 
\[ \hat{t}_k \defequiv  \lfloor (1+\epsilon/C)t_k - M/C \rfloor  \]
The number of slots in the interval $t_k \leq t \leq \hat{t}_k$ 
given by (\ref{eq:t-cond})  is at least: 
\[ \hat{t}_k - t_k + 1 \geq \epsilon t_k/C - M/C  \]
It follows that: 
\[ \frac{1}{\hat{t}_k+1}\sum_{\tau=0}^{\hat{t}_k} \expect{Q(\tau)} \geq M\frac{\epsilon t_k/C - M/C}{(1+\epsilon/C)t_k - M/C+1} \]
Taking a $\limsup$ as $k\rightarrow \infty$ and noting that $\lim_{k\rightarrow\infty} t_k = \infty$
yields: 
\[ \limsup_{k\rightarrow\infty} \frac{1}{\hat{t}_k+1}\sum_{\tau=0}^{\hat{t}_k} \expect{Q(\tau)} \geq \frac{M \epsilon/C}{1+\epsilon/C} \]
This holds for arbitrarily large $M$.  Hence, taking a limit as $M\rightarrow \infty$ yields: 
\[ \limsup_{k\rightarrow\infty} \frac{1}{\hat{t}_k+1}\sum_{\tau=0}^{\hat{t}_k} \expect{Q(\tau)} \geq \infty \]
and thus $Q(t)$ is not strongly stable.  It follows that strongly stable implies mean rate stable. 

A similar proof can be done for the case when $\expect{a(t) + b^{-}(t)} \leq C$ for all $t$.  This 
can be shown by observing that for any $t<t_k$: 
\[   Q(t) \geq Q(t_k) - \sum_{\tau=t_k}^{t-1} [a(\tau) + b^{-}(\tau)]   \]
and hence: 
\[ \expect{Q(t)} \geq \epsilon t_k - C(t-t_k) \]
\end{proof}

\section*{Appendix C --- Proof of Theorem \ref{thm:strong-stability}(b)} 

Here we prove Theorem \ref{thm:strong-stability}(b), which shows that 
strong stability implies rate stability if certain boundedness assumptions are satisfied. 
Suppose $Q(t)$ has dynamics given by (\ref{eq:q-dynamics}), and that there is a 
finite constant $C>0$ such that with probability $1$, we have:
\begin{equation} \label{eq:c-cond} 
b(t) - a(t) \leq C \: \: \forall t \in \{0, 1, 2, \ldots \} 
\end{equation} 
For simplicity, we assume
the condition (\ref{eq:c-cond}) 
holds deterministically (so that we can neglect writing ``with probability 1.'') 
Suppose $Q(t)$ is strongly stable. We want to show that $\lim_{t\rightarrow\infty}Q(t)/t = 0$
with probability $1$.

We prove the result through several preliminary lemmas, presented below. 

\begin{lem} \label{lem:square-root} If $Q(t)$ is strongly stable and if there is a finite 
constant $C>0$ such that (\ref{eq:c-cond}) holds for all $t$, 
then $\expect{Q(t)/t} \leq O(1/\sqrt{t})$. Specifically, 
there exists a finite constant $D>0$ and a 
positive timeslot $t_D$ such that
\begin{equation} \label{eq:srthm1} 
\expect{Q(t)/t} \leq D/\sqrt{t} \: \: \mbox{ for all $t \geq t_D$} 
\end{equation} 
  Hence, for all $t \geq t_D$ and all $\epsilon>0$  we have: 
\begin{equation} \label{eq:srthm2} 
 Pr[Q(t)/t \geq \epsilon/4] \leq  4D/(\epsilon\sqrt{t}) 
 \end{equation} 
\end{lem}

\begin{proof} 
Because $Q(t)$ is strongly stable, there is a finite constant $B>0$ such that: 
\[ \limsup_{t\rightarrow\infty} \frac{1}{t} \sum_{\tau=0}^{t-1} \expect{Q(\tau)} < B < \infty \]
Then for large $t$ (all $t \geq t^*$ for some $t^*$),  we have: 
\[ \frac{1}{t}\sum_{\tau=0}^{t-1} \expect{Q(\tau)} \leq B \]

Suppose now that for any finite $D>0$, 
there exist arbitrarily large times $t_i$ such that $\expect{Q(t_i)/t_i} > D/\sqrt{t_i}$. 
We shall reach a contradiction.  For $t_i \geq t^*$ we have: 
\begin{eqnarray}
B \geq \frac{1}{2t_i}\sum_{\tau=0}^{2t_i-1} \expect{Q(\tau)} &\geq& \frac{1}{2}\sum_{\tau=t_i}^{2t_i-1} \expect{\frac{Q(\tau)}{t_i}} \nonumber \\
&\geq& \frac{1}{2}\sum_{\tau=t_i}^{2t_i-1} \expect{\frac{Q(\tau)}{\tau}}  \label{eq:foop}
\end{eqnarray}
Because (\ref{eq:c-cond}) holds, we 
know that  for all $\tau \geq t_i$:
\[ Q(\tau) \geq Q(t_i) - C(\tau - t_i) \]
Further,  $\expect{Q(t_i)/t_i} > D/\sqrt{t_i}$ and so $\expect{Q(t_i)} > D\sqrt{t_i}$.  
We thus have for all $\tau \in \{t_i, \ldots, 2t_i - 1\}$): 
\[ \frac{\expect{Q(\tau)}}{\tau} > \frac{D\sqrt{t_i} -C(\tau-t_i)}{\tau}  \]
Now assume that $t_i$ is large, so that $2t_i - 1 \geq  t_i + \lfloor D\sqrt{t_i}/(2C) \rfloor \geq t_i + 
D\sqrt{t_i}/(4C)$. 
Note that if: 
 \[ \tau \in\{t_i, \ldots, t_i + \lfloor D\sqrt{t_i}/(2C)\rfloor\} \]
 then $\tau - t_i \leq D\sqrt{t_i}/(2C)$ and we know: 
\begin{eqnarray*}
 \frac{\expect{Q(\tau)}}{\tau} &>& \frac{D\sqrt{t_i} - C(\tau-t_i)}{\tau} \\
 &\geq& \frac{D\sqrt{t_i} - C(\tau-t_i)}{2t_i} \\
 &\geq& D/(4\sqrt{t_i}) 
\end{eqnarray*}
 It follows that:
\[ \sum_{\tau=t_i}^{2t_i-1} \expect{Q(\tau)/\tau} \geq \left(\frac{D}{4\sqrt{t_i}}\right) \frac{D\sqrt{t_i}}{4C} 
= \frac{D^2}{16C} \]
Therefore, for large $t_i$, from (\ref{eq:foop})  we have:  
\[    B \geq  \frac{D^2}{32C} \]
The above inequality must hold for all $D>0$.  This clearly does not hold for $D> \sqrt{32 BC}$, yielding
a contradiction and hence proving the result (\ref{eq:srthm1}). 

To prove (\ref{eq:srthm2}), we use the fact that (\ref{eq:srthm1}) holds to get for any time
$t > t_D$: 
\begin{eqnarray*}
D/\sqrt{t} \geq \expect{Q(t)/t} \geq  (\epsilon/4)Pr[Q(t)/t \geq \epsilon/4]
\end{eqnarray*}
Dividing the above by $\epsilon/4$ yields (\ref{eq:srthm2}). 
\end{proof} 

Now again suppose $a(t)$ and $b(t)$ satisfies (\ref{eq:c-cond}) for all $t$. 
Fix $\epsilon>0$ and define a constant $\alpha >0$ as follows: 
\[ \alpha \defequiv \min\left[\frac{\epsilon}{2C}, \frac{1}{2} \right] \]
Fix an integer time $t_0$ such that $t_0 \geq 1/\alpha$, and define the following 
sequence for $i \in \{0, 1, 2, \ldots\}$:
\[ t_{i+1} = t_i + \lfloor \alpha t_i \rfloor \]
Note that for all $i$ we have $t_{i+1} > t_i$ (because $\alpha t_i \geq 1$). 
Let the set of times in the interval $\{t_i, \ldots, t_i + \lfloor \alpha t_i \rfloor  - 1\}$ denote the $i$th frame.

\begin{lem}  If $Q(t)/t \geq \epsilon$ for any time $t\geq t_0$ 
in the $i$th frame (for some $i \in \{0, 1, 2, \ldots\}$), then $Q(t_{i+1})/t_{i+1} \geq \epsilon/4$. 
\end{lem} 

\begin{proof} Fix $\epsilon>0$, and 
suppose $Q(t)/t \geq \epsilon$, where $t\geq t_0$ and $t$ is in the $i$th frame for some $i \in \{0, 1, 2, \ldots\}$. 
Then: 
\begin{equation} \label{eq:ti} 
 t_i \leq t <  t_{i+1} \leq t_i +\alpha t_i 
 \end{equation} 
Hence: 
\begin{equation} \label{eq:lemgeom} 
  t_{i+1} -t \leq \alpha t_i
  \end{equation} 
However: 
\[ Q(t_{i+1}) \geq Q(t) -C (t_{i+1} - t) \]
Therefore (because $Q(t)/t \geq \epsilon$) we have: 
\begin{eqnarray}
 Q(t_{i+1}) &\geq& \epsilon t - C (t_{i+1} - t)  \nonumber \\
 &\geq&  \epsilon t_i  - C(t_{i+1} - t) \label{eq:ti1} \\
 &\geq& \epsilon t_i - C \alpha t_i \label{eq:ti2} \\
 &\geq& t_i\epsilon/2 \label{eq:ti3} 
 \end{eqnarray} 
 where (\ref{eq:ti1}) follows because $t \geq t_i$, (\ref{eq:ti2}) follows from (\ref{eq:lemgeom}), 
 and (\ref{eq:ti3}) follows because $\alpha \leq \epsilon/(2C)$ (by definition of $\alpha$). 
 Thus: 
 \[  \frac{Q(t_{i+1})}{t_{i+1}} \geq \frac{\epsilon}{2}\left(\frac{t_i}{t_{i+1}}\right) \]
 However, from (\ref{eq:ti})  we have: 
 \[ t_{i+1} - t_i \leq \alpha t_i \]
 and so: 
 \[ 1 - (t_i/t_{i+1}) \leq \alpha (t_i/t_{i+1})  \leq \alpha \]
 Thus: 
 \[ \frac{t_i}{t_{i+1}} \geq  1 - \alpha  \geq  1/2 \]
 where the last inequality follows because $\alpha \leq 1/2$ (by definition of $\alpha$). 
 Using this in (\ref{eq:ti3}) yields: 
 \[ \frac{Q(t_{i+1})}{t_{i+1}} \geq \left(\frac{t_i}{t_{i+1}}\right)\epsilon/2  \geq \epsilon/4 \]
This proves the result. 
\end{proof} 

\begin{lem} For each frame $i \in \{0, 1, 2, \ldots\}$ we have: 
\begin{eqnarray*}
Pr[\mbox{$Q(t)/t \geq \epsilon$  for some $t$ in the $i$th frame}] \\
\leq Pr[Q(t_{i+1})/t_{i+1} \geq \epsilon/4] 
\end{eqnarray*}
\end{lem} 
\begin{proof} 
This follows immediately from the previous Lemma. 
\end{proof} 

The remainder of the proof is similar to the standard 
proof of the strong law of large numbers
(see, for example, \cite{billingsley}). 
We have for any $t_0$ that starts the frames: 
\begin{eqnarray}
 && \hspace{-.3in} Pr\left[\limsup_{t\rightarrow\infty} Q(t)/t \geq \epsilon\right] \nonumber \\
 &\leq& Pr\left[\sup_{t\geq t_0} Q(t)/t \geq \epsilon\right] \nonumber \\
 &\leq& \sum_{i=0}^{\infty} Pr[\mbox{$Q(\tau)/\tau \geq \epsilon$ for some $\tau$ in the $i$th frame}]\nonumber \\
 &\leq& \sum_{i=0}^{\infty} Pr[Q(t_{i+1})/t_{i+1} \geq \epsilon/4] \label{eq:final-summ} 
\end{eqnarray}

However, $t_i$ is an exponentially growing sequence (note that $t_{i+1} \geq (1+\alpha)t_i - 1$), and so we are sampling $Q(t)/t$ at exponentially
increasing times. However, if $Q(t)$ is strongly stable, 
from Lemma \ref{lem:square-root} we know that $Pr[Q(t)/t \geq \epsilon/4] \leq O(1/\sqrt{t})$. 
Hence, the final
sum in (\ref{eq:final-summ}) is summable and goes to zero as $t_0 \rightarrow \infty$. 
This proves that if $Q(t)$ is strongly stable, then for any $\epsilon>0$ we know: 
\[ Pr\left[\limsup_{t\rightarrow\infty} Q(t)/t \geq \epsilon\right] = 0\]
Because this holds for all $\epsilon>0$, it must be the case that $Q(t)/t \rightarrow 0$ with probability 1 (and so 
the queue is rate stable). 

This proof considers the case when $b(t) - a(t) \leq C$ for all $t$.  The other case when $a(t) + b^-(t) \leq C$
for all $t$ is proven similarly and is omitted for brevity.

\section*{Appendix D --- Proof of Theorem \ref{thm:ss-implies-mrs}}

Here we prove Theorem \ref{thm:ss-implies-mrs}.
Suppose there is a finite constant $C$ such that $a(t) + b^-(t) \leq C$ with probability $1$
for all $t$.  For simplicity, we assume that $Q(0) = 0$, and that $a(t) + b^-(t) \leq C$ deterministically
(so that we do not need to repeat the phrase ``with probability 1''). 
Suppose that $Q(t)$ is \emph{not} mean rate stable.  We show that it is not
steady state stable.

Because $Q(t)$ is not mean rate stable, there must be an $\epsilon>0$ and an infinite
collection of increasing times $t_k$ such 
that $\expect{Q(t_k)/t_k} \geq \epsilon$ for 
all $k\in\{1,2, \ldots\}$ and $\lim_{k\rightarrow\infty} t_k = \infty$. 
Now fix an (arbitrarily large) value $M$.   We have for any time $t \leq t_k$: 
\[ Q(t_k) \leq Q(t) + \sum_{\tau=t}^{t_k-1} [a(\tau) + b^-(\tau)]  \leq Q(t) + C(t_k-t) \]
Thus: 
\begin{equation} 
\expect{Q(t)} \geq \expect{Q(t_k)} - C(t_k - t) \geq \epsilon t_k - C(t_k-t) \label{eq:tm5} 
\end{equation} 
On the other hand, for $t \leq t_k$ we have: 
\begin{eqnarray*}
\expect{Q(t)} &\leq& MPr[Q(t)\leq M] \\
&& + \expect{Q(t)|Q(t)>M}Pr[Q(t)>M] \\
&\leq& M + CtPr[Q(t)>M]  \\
&\leq& M  + Ct_kPr[Q(t)>M] 
\end{eqnarray*}
where we have used the fact that $Q(t) \leq Ct \leq Ct_k$ (because $Q(0) = 0$, the queue increases
by at most $C$ on each slot, and $t \leq t_k$). 
Thus: 
\[ \expect{Q(t)} \leq M + Ct_kPr[Q(t)>M] \]
Combining this with (\ref{eq:tm5}) yields: 
\[ M + Ct_kPr[Q(t)>M] \geq \epsilon t_k - C(t_k-t) \]
Therefore, for all $t \leq t_k$ we have: 
\begin{equation} \label{eq:t5a}  
Pr[Q(t)>M] \geq \frac{\epsilon t_k - C(t_k-t) - M}{Ct_k} 
\end{equation}
Now suppose that: 
\begin{equation} \label{eq:t-cond2} 
0 \leq (t_k-t) \leq (\epsilon/2C)t_k 
\end{equation} 
It follows from (\ref{eq:t-cond2})  that:
\[ C(t_k-t) \leq (\epsilon/2) t_k \]
Using this in (\ref{eq:t5a}) gives: 
\[ Pr[Q(t)>M] \geq \frac{(\epsilon/2)t_k - M}{Ct_k} \]
The above holds for all $t$ that satisfy (\ref{eq:t-cond2}). 
Now assume that $k$ is large enough to ensure that $(\epsilon/2)t_k - M \geq (\epsilon/4)t_k$
(this is true for sufficiently large $k$ because $\lim_{k\rightarrow\infty} t_k = \infty$.
Thus, for sufficiently large $k$, and if (\ref{eq:t-cond2}) holds, the above bound becomes: 
\begin{equation} \label{eq:t5b} 
Pr[Q(t)>M] \geq \frac{\epsilon}{4C} 
\end{equation} 
From (\ref{eq:t-cond2}), we see that the number of slots $t\leq t_k$ for which 
(\ref{eq:t5b}) holds is at least $[\epsilon/(2C)]t_k$. Therefore: 
\[ \frac{1}{t_k+1} \sum_{\tau=0}^{t_k} Pr[Q(\tau)>M] \geq \frac{1}{t_k+1}\left(\frac{\epsilon}{4C}\right)\left(\frac{t_k\epsilon}{2C}\right) \]
It follows that: 
\[ \limsup_{k\rightarrow\infty}\frac{1}{t_k+1}\sum_{\tau=0}^{t_k} Pr[Q(\tau)>M] \geq \frac{\epsilon^2}{8C^2} \]
Thefore, the function $g(M)$ defined by (\ref{eq:gm}) satisfies for all $M>0$: 
\[ g(M) \geq \frac{\epsilon^2}{8C^2} \]
It follows that: 
\[ \lim_{M\rightarrow\infty} g(M) \geq \frac{\epsilon^2}{8C^2} > 0 \]
and hence $Q(t)$ is not steady state stable.

\section*{Appendix E --- Additional Sample Path Results} 

We begin by reviewing 
some general results concerning expectations of limits of non-negative random processes.   
These can be viewed as probabilistic
interpretations of Fatou's Lemma and the Lebesgue Dominated Convergence Theorem from measure theory (where
an expectation can be viewed as an integral over an appropriate probability measure). 
We state these without proof (see, for example, \cite{billingsley}). 
Both lemmas below are stated for 
a general non-negative stochastic process $X(t)$ 
defined over $t \geq 0$ (where $t$ can either be an  integer index
or a value in the  real number line). 

\begin{lem} \label{lem:fatou} (Fatou's Lemma)  For any non-negative stochastic process $X(t)$ for $t\geq 0$, we have: 
\[ \liminf_{t\rightarrow\infty} \expect{X(t)} \geq \expect{\liminf_{t\rightarrow\infty} X(t)} \]
where both both sides of the inequality can be potentially infinite. 
\end{lem} 

\begin{lem} \label{lem:lebesgue} (Lebesgue Dominated Convergence Theorem) Suppose that $X(t)$ is a non-negative
stochastic process for $t \geq 0$, and that there is a non-negative random variable $Y$, defined on the same probability space
as the process $X(t)$,  such that $X(t) \leq Y$ (with probability $1$) 
for all $t$. Further assume that $\expect{Y} < \infty$.  Then: 

(a) $\limsup_{t\rightarrow\infty} \expect{X(t)} \leq \expect{\limsup_{t\rightarrow\infty} X(t)} \leq \expect{Y}$ 

(b) In addition to the assumption that $X(t) \leq Y$ with probability 1, if 
$X(t)$ converges to a non-negative random variable $X$ with probability 1, then the limit
of $\expect{X(t)}$ is well defined, and: 
\[ \lim_{t\rightarrow\infty} \expect{X(t)} = \expect{X} \] 
\end{lem}

\subsection{Applications to Queue Sample Path Analysis} 

\begin{lem} Let $Q(t)$ be a general non-negative random process defined over $t \in \{0, 1, 2, \ldots\}$. 
Suppose that the following limit is well defined as a (possibly infinite) random variable $Q_{av}$ with probability 1: 
\begin{equation} 
 \lim_{t\rightarrow\infty} \frac{1}{t}\sum_{\tau=0}^{t-1} Q(\tau) = Q_{av} \: \: \mbox{ with probability 1}   \label{eq:sample-path-strong}
\end{equation} 
Then: 
\begin{equation*} 
 \expect{Q_{av}} \leq \liminf_{t\rightarrow\infty} \frac{1}{t}\sum_{\tau=0}^{t-1} \expect{Q(\tau)}
 \end{equation*} 
Therefore, if $Q(t)$ is strongly stable, it must be that $Q_{av} < \infty$ with probability 1. 
\end{lem} 
\begin{proof} 
Define $X(t) = \frac{1}{t}\sum_{\tau=0}^{t-1}Q(\tau)$, and note that $\liminf_{t\rightarrow\infty} X(t) = Q_{av}$ with probability 1. 
The result then follows as an immediate consequence of Lemma \ref{lem:fatou}. 
\end{proof} 

Using standard Markov chain theory and renewal theory, it can be shown that the limit in (\ref{eq:sample-path-strong}) 
exists as a (possibly infinite) random variable $Q_{av}$ with probability 1 whenever $Q(t)$ evolves according to a 
Markov chain such that, with probability 1, there is at least one state that is visited infinitely often
with finite mean recurrence times (not necessarily the
same state on each sample path realization).  This holds even if the Markov chain is not irreducible, not aperiodic, and/or
has an uncountably infinite state space.
The converse statement
does not hold:  Note that if 
the limit (\ref{eq:sample-path-strong}) holds with $Q_{av} < \infty$ with probability 1, this does \emph{not} always imply
that $Q(t)$ is strongly stable. The same counter-example from Subsection \ref{subsection:counterexamples-rate-mean} can be
used to illustrate this. 

\begin{lem} Let $Q(t)$ be a general non-negative random process defined over $t \in \{0, 1, 2, \ldots\}$. 
Suppose that for all $M>0$, the following limit exists as a random variable $h(M)$ with probability $1$: 
\begin{equation} \label{eq:sample-path-steady-state} 
 \lim_{t\rightarrow\infty} \frac{1}{t}\sum_{\tau=0}^{t-1} 1\{Q(\tau)> M\} = h(M) \: \: \mbox{ with probability 1}
 \end{equation} 
where $1\{Q(\tau)>M\}$ is an indicator function that is $1$ whenever $Q(\tau)>M$, and zero otherwise. 
Then the following limit for $g(M)$ is well defined: 
\[ \lim_{t\rightarrow\infty} \frac{1}{t}\sum_{\tau=0}^{t-1} Pr[Q(\tau) > M] = g(M) \]
Furthermore: 
\[ g(M) = \expect{h(M)} \]
It follows that if $Q(t)$ is steady state stable (so that $g(M) \rightarrow 0$ as $M \rightarrow \infty$), then: 
\[ \lim_{t\rightarrow\infty} h(M) = 0 \: \: \mbox{ with probability 1} \]
\end{lem} 
\begin{proof} 
The result can be shown  by application of Lemma \ref{lem:lebesgue}, using $X(t) = \frac{1}{t}\sum_{\tau=0}^{t-1} 1\{Q(\tau)>M\}$
and $Y=1$.  
\end{proof} 

By basic renewal theory and Markov chain theory, it can be shown that the limit in (\ref{eq:sample-path-steady-state}) 
holds whenever $Q(t)$ evolves according to a Markov chain (possibly non-irreducible, non-aperiodic) and such that either 
the event $\{Q(t)>M\}$ or its complement $\{Q(t) \leq M\}$ can be written as the union of a finite number of states. 

\section*{Appendix F --- Proof of Lemma \ref{lem:drift2}} 

\begin{proof} (Lemma \ref{lem:drift2})
By definition of a $C$-approximate decision, at every slot 
 $\tau \in \{t, \ldots, t+T-1\}$ we have: 
 \begin{eqnarray*}
 V\expect{f(\hat{\bv{x}}(\alpha(\tau), \omega(\tau)))|\bv{\Theta}(t)} \\
 + \sum_{l=1}^L\expect{Z_l(\tau)g_l(\hat{\bv{x}}(\alpha(\tau), \omega(\tau))) |\bv{\Theta}(t)}\\
 + \sum_{k=1}^K\expect{Q_k(\tau)a_k(\tau)|\bv{\Theta}(t)} \\
  + \sum_{k=1}^K\expect{Q_k(\tau)[\hat{y}_k(\alpha(\tau), \omega(\tau)) - \hat{b}_k(\alpha(\tau), \omega(\tau))]|\bv{\Theta}(t)} \\
  \leq  C + 
 V\expect{f(\hat{\bv{x}}(\alpha^*(\tau), \omega(\tau)))|\bv{\Theta}(t)} \\
  + \sum_{l=1}^L\expect{Z_l(\tau)g_l(\hat{\bv{x}}(\alpha^*(\tau), \omega(\tau))) |\bv{\Theta}(t)}\\
   + \sum_{k=1}^K\expect{Q_k(\tau)a_k(\tau)|\bv{\Theta}(t)} \\
 +  \sum_{k=1}^K\expect{Q_k(\tau)[\hat{y}_k(\alpha^*(\tau), \omega(\tau)) - \hat{b}_k(\alpha^*(\tau), \omega(\tau))]|\bv{\Theta}(t)} 
 \end{eqnarray*}
 where $\alpha^*(\tau)$ is any other (possibly randomized) decision in $\script{A}_{\omega(\tau)}$. 

 However, we have for all $\tau \in \{t, \ldots, t+T-1\}$:  
 \begin{eqnarray*}
 |Z_l(\tau) - Z_l(t)| &\leq& \sum_{v=t}^{\tau-1} |g_l(\bv{x}(v))|  \\
 |Q_k(\tau) - Q_k(t)| &\leq& \sum_{v=t}^{\tau-1} [y_k(v) + a_k(v) + b_k(v)] 
 \end{eqnarray*}
 Thus: 
 \begin{eqnarray*}
 V\expect{f(\hat{\bv{x}}(\alpha(\tau), \omega(\tau)))|\bv{\Theta}(t)} \\
 + \sum_{l=1}^L\expect{Z_l(t)g_l(\hat{\bv{x}}(\alpha(\tau), \omega(\tau))) |\bv{\Theta}(t)}\\
    + \sum_{k=1}^K\expect{Q_k(t)a_k(\tau)|\bv{\Theta}(t)} \\
  + \sum_{k=1}^K\expect{Q_k(t)[\hat{y}_k(\alpha(\tau), \omega(\tau)) - \hat{b}_k(\alpha(\tau), \omega(\tau))]|\bv{\Theta}(t)} \\
  \leq  C + (\tau-t)2\expect{\hat{D} |\bv{\Theta}(t)}   \\
 + V\expect{f(\hat{\bv{x}}(\alpha^*(\tau), \omega(\tau)))|\bv{\Theta}(t)} \\
  + \sum_{l=1}^L\expect{Z_l(t)g_l(\hat{\bv{x}}(\alpha^*(\tau), \omega(\tau))) |\bv{\Theta}(t)}\\
     + \sum_{k=1}^K\expect{Q_k(t)a_k(\tau)|\bv{\Theta}(t)} \\
 +  \sum_{k=1}^K\expect{Q_k(t)[\hat{y}_k(\alpha^*(\tau), \omega(\tau)) - \hat{b}_k(\alpha^*(\tau), \omega(\tau))]|\bv{\Theta}(t)} 
 \end{eqnarray*}
 where $2\hat{D}$ is a random variable, with $\expect{\hat{D}} \leq D$, where $D$ is related
 (via Cauchy-Schwartz) 
 to the worst case second moments of 
 $g_l(\bv{x}(t))$, $a_k(t)$, $b_k(t)$, $y_k(t)$.  A more detailed description of $D$ is given at the 
 end of this subsection. 
 Summing over $\tau \in \{t, \ldots, t+T-1\}$ yields:  
  \begin{eqnarray*}
 V\expect{\sum_{\tau=0}^{t+T-1}f(\hat{\bv{x}}(\alpha(\tau), \omega(\tau)))|\bv{\Theta}(t)} \\
 + \sum_{l=1}^L\sum_{\tau=0}^{t+T-1}\expect{Z_l(t)g_l(\hat{\bv{x}}(\alpha(\tau), \omega(\tau))) |\bv{\Theta}(t)}\\
  + \sum_{k=1}^K\sum_{\tau=0}^{t+T-1}\mathbb{E}\left\{Q_k(t)[\hat{y}_k(\alpha(\tau), \omega(\tau)) 
  - \right. \\
  \left.  \hat{b}_k(\alpha(\tau), \omega(\tau))]|\bv{\Theta}(t) \right\} \\
  \leq  CT + T(T-1)\expect{\hat{D}|\bv{\Theta}(t)} \\
+ V\sum_{\tau=0}^{t+T-1}\expect{f(\hat{\bv{x}}(\alpha^*(\tau), \omega(\tau)))|\bv{\Theta}(t)} \\
  + \sum_{l=1}^L\sum_{\tau=0}^{t+T-1}\expect{Z_l(t)g_l(\hat{\bv{x}}(\alpha^*(\tau), \omega(\tau))) |\bv{\Theta}(t)}\\
 +  \sum_{k=1}^K\sum_{\tau=0}^{t+T-1}\mathbb{E}\left\{Q_k(t)[\hat{y}_k(\alpha^*(\tau), \omega(\tau)) - \right. \\\left.\hat{b}_k(\alpha^*(\tau), \omega(\tau))]|\bv{\Theta}(t)\right\} 
 \end{eqnarray*}
 
 The above inequality is an upper bound for  the right hand side of (\ref{eq:drift}), which proves the result
 of (\ref{eq:drift2}).  
 \end{proof} 
 
 For more details on $\hat{D}$ and $D$, we note that $2\hat{D}$ can be defined to be $0$ if 
 $\tau = t$, and if $\tau > t$ it is defined: 
 \begin{eqnarray*}
2\hat{D} \defequiv  \sum_{k=1}^K[\hat{y}_k(\alpha^*(\tau), \omega(\tau)) +  \hat{b}_k(\alpha^*(\tau), \omega(\tau))]\times \\
\left[ \frac{1}{\tau -t}\sum_{v=t}^{\tau-1}[y_k(v)+a_k(v)+b_k(v)]\right] \\
+  \sum_{k=1}^K[\hat{y}_k(\alpha(\tau), \omega(\tau)) + \hat{b}_k(\alpha(\tau), \omega(\tau))]\times \\
\left[ \frac{1}{\tau-t}\sum_{v=t}^{\tau-1}[y_k(v)+a_k(v)+b_k(v)]\right] \\
+ \sum_{l=1}^L|g_l(\hat{\bv{x}}(\alpha^*(\tau), \omega(\tau)))|\frac{1}{\tau-t}\sum_{v=t}^{\tau-1}|g_l(\bv{x}(v))|  \\
+ \sum_{l=1}^L|g_l(\hat{\bv{x}}(\alpha(\tau), \omega(\tau)))|\frac{1}{\tau-t}\sum_{v=t}^{\tau-1}|g_l(\bv{x}(v))|  
\end{eqnarray*}
By the Cauchy-Schwartz inequality and Jensen's inequality, it follows that $2\expect{\hat{D}} \leq 2D$, 
where $D$ is a finite constant that satisfies: 
\begin{eqnarray*}
D \geq \sum_{k=1}^K \expect{(\hat{y}_k(\alpha', \omega) + a_k(0)  + \hat{b}_k(\alpha',\omega))^2} \\
+ \sum_{l=1}^L\expect{g_l(\hat{\bv{x}}(\alpha'', \omega))^2} 
\end{eqnarray*}
where the first term represents the worst case second moment of $\hat{y}_k(\cdot)+ a_k(0)
+ \hat{b}_k(\cdot)$ over any decision $\alpha'$ (where the expectation is with respect to the stationary
distribution $\pi(\omega)$ for $\omega$), and the second term is the worst case second moment
of $g_l(\hat{\bv{x}}(\cdot))$. 

 \section*{Appendix G --- Proof of Lemma \ref{lem:driftyy}} 
 
 \begin{proof} (Lemma \ref{lem:driftyy}) 
 From (\ref{eq:multi-q-dynamics}), it can be shown that (see \cite{now}): 
 \begin{eqnarray*}
 Q_k(t+T) \leq \max\left[Q_k(t) - \sum_{\tau=t}^{t+T-1} b_k(\tau), 0\right] \\
 + \sum_{\tau=t}^{t+T-1}[a_k(\tau)  + y_k(\tau)] 
 \end{eqnarray*}
 Using the fact that for $Q\geq0$, $\mu\geq 0$, $a\geq 0$: 
 \[ \frac{1}{2}(\max[Q-\mu,0] + a)^2 \leq \frac{Q^2 + \mu^2  + a^2}{2} + Q(a - \mu)  \]
 yields: 
 \begin{eqnarray*}
 \frac{Q_k(t+T)^2}{2} \leq \frac{Q_k(t)^2}{2} + \frac{1}{2}\left(\sum_{\tau=t}^{t+T-1} b_k(\tau)\right)^2  \\
 + \frac{1}{2}\left(  \sum_{\tau=t}^{t+T-1}[a_k(\tau)  + y_k(\tau)] \right)^2 \\
 + Q_k(t)\left( \sum_{\tau=t}^{t+T-1}[a_k(\tau) + y_k(\tau) - b_k(\tau)]   \right)
 \end{eqnarray*}
 Similarly, from (\ref{eq:z-dynamics}) it can be shown that: 
 \[ Z_l(t+T) \leq \max\left[Z_l(t) + 
 \sum_{\tau=t}^{t+T-1}g_l(\bv{x}(\tau)) , \sum_{\tau=t}^{t+T-1}|g_l(\bv{x}(\tau)|\right] \]
 and so: 
  \begin{eqnarray*}
    \frac{Z_l(t+T)^2}{2} \leq \frac{Z_l(t)^2}{2} + \left(\sum_{\tau=t}^{t+T-1}|g_l(\bv{x}(\tau))|\right)^2 \\
    + Z_l(t)\sum_{\tau=t}^{t+T-1}g_l(\bv{x}(\tau))   
    \end{eqnarray*}
 Combining these, summing, and taking conditional expectations yields: 
 \begin{eqnarray*}
 \Delta_T(\bv{\Theta}(t)) \leq T^2\expect{\hat{B}|\bv{\Theta}(t)}   \\
 + \sum_{k=1}^KQ_k(t)\sum_{\tau=t}^{t+T-1}\expect{a_k(\tau) + y_k(\tau) - b_k(\tau)|\bv{\Theta}(t)}\\
 + \sum_{l=1}^LZ_l(t)\sum_{\tau=t}^{t+T-1} \expect{g_l(\bv{x}(\tau))|\bv{\Theta}(t)} 
 \end{eqnarray*}
 where $\hat{B}$ is a random variable that satisfies: 
 \begin{eqnarray*}
 \expect{\hat{B}|\bv{\Theta}(t)}  = \frac{1}{2}\sum_{k=1}^K\expect{\left(\frac{1}{T}\sum_{\tau=t}^{t+T-1}b_k(\tau)\right)^2|\bv{\Theta}(t)} \\
+   \frac{1}{2}\sum_{k=1}^K\expect{\left(\frac{1}{T}\sum_{\tau=t}^{t+T-1}[a_k(\tau)+y_k(\tau)]\right)^2|\bv{\Theta}(t)} \\
+ \sum_{l=1}^L\expect{\left(\frac{1}{T}\sum_{\tau=t}^{t+T-1}|g_l(\bv{x}(\tau))|\right)^2|\bv{\Theta}(t)} 
 \end{eqnarray*} 
 By  Jensen's inequality, 
 we have that $\expect{\hat{B}} \leq B$, where $B$ is a 
 constant that satisfies for all $t$: 
 \begin{eqnarray*}
 B \geq \frac{1}{2}\sum_{k=1}^K\expect{b_k(t)^2} + \frac{1}{2}\sum_{k=1}^K\expect{(a_k(t)+y_k(t))^2} \\
 + \sum_{l=1}^L\expect{g_l(\bv{x}(t))^2} 
 \end{eqnarray*}
 Such a finite constant exists by the second moment boundedness assumptions. 
 The result of (\ref{eq:drift}) follows by adding the following term to both sides: 
 \[ V\sum_{\tau=t}^{t+T-1} \expect{f(\bv{x}(\tau))|\bv{\Theta}(t)} \]
 \end{proof}

\bibliographystyle{unsrt}
\bibliography{../../latex-mit/bibliography/refs}

\begin{thebibliography}{10}

\bibitem{baccelli-book}
F.~Baccelli and P.~Br\'{e}maud.
\newblock {\em Elements of Queueing Theory}.
\newblock Berlin: Springer, 2nd Edition, 2003.

\bibitem{queue-mazumdar}
F.~M. Guillemin and R.~R. Mazumdar.
\newblock On pathwise analysis and existence of empirical distributions for
  g/g/1 queues.
\newblock {\em Stoc. Proc. Appl.}, vol. 67 (1), 1997.

\bibitem{bertsekas-data-nets}
D.~P. Bertsekas and R.~Gallager.
\newblock {\em Data Networks}.
\newblock New Jersey: Prentice-Hall, Inc., 1992.

\bibitem{asmussen-prob}
S$\mbox{\o}$ren Asmussen.
\newblock {\em Applied Probability and Queues, Second Edition}.
\newblock New York: Spring-Verlag, 2003.

\bibitem{foss-stability}
S.~Foss and T.~Konstantopoulos.
\newblock An overview of some stochastic stability methods.
\newblock {\em Journal of Operation Research Society Japan}, vol. 47, no. 4,
  pp. 275-303, 2004.

\bibitem{kumar-meyn-stability}
P.~R. Kumar and S.~P. Meyn.
\newblock Stability of queueing networks and scheduling policies.
\newblock {\em IEEE Trans. on Automatic Control}, vol.40,.n.2, pp.251-260, Feb.
  1995.

\bibitem{now}
L.~Georgiadis, M.~J. Neely, and L.~Tassiulas.
\newblock Resource allocation and cross-layer control in wireless networks.
\newblock {\em Foundations and Trends in Networking}, vol. 1, no. 1, pp. 1-149,
  2006.

\bibitem{dai-fluid}
J.~G. Dai.
\newblock On positive harris recurrence of multiclass queueing networks: a
  unified approach via fluid limit models.
\newblock {\em Annals of Applied Probability}, vol. 5, pp. 49-77, 1995.

\bibitem{dai-balaji-lyap}
J.~G. Dai and B.~Prabhakar.
\newblock The throughput of data switches with and without speedup.
\newblock {\em Proc. IEEE INFOCOM}, 2000.

\bibitem{neely-thesis}
M.~J. Neely.
\newblock {\em Dynamic Power Allocation and Routing for Satellite and Wireless
  Networks with Time Varying Channels}.
\newblock PhD thesis, Massachusetts Institute of Technology, LIDS, 2003.

\bibitem{neely-energy-it}
M.~J. Neely.
\newblock Energy optimal control for time varying wireless networks.
\newblock {\em IEEE Transactions on Information Theory}, vol. 52, no. 7, pp.
  2915-2934, July 2006.

\bibitem{neely-fairness-infocom05}
M.~J. Neely, E.~Modiano, and C.~Li.
\newblock Fairness and optimal stochastic control for heterogeneous networks.
\newblock {\em Proc. IEEE INFOCOM}, March 2005.

\bibitem{tass-radio-nets}
L.~Tassiulas and A.~Ephremides.
\newblock Stability properties of constrained queueing systems and scheduling
  policies for maximum throughput in multihop radio networks.
\newblock {\em IEEE Transacations on Automatic Control}, vol. 37, no. 12, pp.
  1936-1949, Dec. 1992.

\bibitem{tass-server-allocation}
L.~Tassiulas and A.~Ephremides.
\newblock Dynamic server allocation to parallel queues with randomly varying
  connectivity.
\newblock {\em IEEE Transactions on Information Theory}, vol. 39, pp. 466-478,
  March 1993.

\bibitem{tass-one-hop}
L.~Tassiulas.
\newblock Scheduling and performance limits of networks with constantly
  changing topology.
\newblock {\em IEEE Trans. on Inf. Theory}, May 1997.

\bibitem{neely-power-network-jsac}
M.~J. Neely, E.~Modiano, and C.~E Rohrs.
\newblock Dynamic power allocation and routing for time varying wireless
  networks.
\newblock {\em IEEE Journal on Selected Areas in Communications}, vol. 23, no.
  1, pp. 89-103, January 2005.

\bibitem{rahul-cognitive-tmc}
R.~Urgaonkar and M.~J. Neely.
\newblock Opportunistic scheduling with reliability guarantees in cognitive
  radio networks.
\newblock {\em IEEE Transactions on Mobile Computing}, vol. 8, no. 6, pp.
  766-777, June 2009.

\bibitem{longbo-profit-allerton07}
L.~Huang and M.~J. Neely.
\newblock The optimality of two prices: Maximizing revenue in a stochastic
  network.
\newblock {\em Proc. 45th Allerton Conf. on Communication, Control, and
  Computing}, Sept. 2007.

\bibitem{neely-maximal-bursty-ton}
M.~J. Neely.
\newblock Delay analysis for maximal scheduling with flow control in wireless
  networks with bursty traffic.
\newblock {\em IEEE Transactions on Networking}, vol. 17, no. 4, pp. 1146-1159,
  August 2009.

\bibitem{neely-mesh}
M.~J. Neely and R.~Urgaonkar.
\newblock Cross layer adaptive control for wireless mesh networks.
\newblock {\em Ad Hoc Networks (Elsevier)}, vol. 5, no. 6, pp. 719-743, August
  2007.

\bibitem{neely-universal-scheduling}
M.~J. Neely.
\newblock Universal scheduling for networks with arbitrary traffic, channels,
  and mobility.
\newblock {\em ArXiv technical report}, arXiv:1001.0960v1, Jan. 2010.

\bibitem{neely-mwl-ita}
M.~J. Neely.
\newblock Max weight learning algorithms with application to scheduling in
  unknown environments.
\newblock {\em Information Theory and Applications Workshop (ITA), San Diego,
  invited paper}, Feb. 2009.

\bibitem{ross}
S.~Ross.
\newblock {\em Stochastic Processes}.
\newblock John Wiley \& Sons, Inc., New York, 1996.

\bibitem{sample-path-queue}
M.~El-Taha and S.~Stidham Jr.
\newblock {\em Sample-Path Analysis of Queueing Systems}.
\newblock Kluwer Academic Publishers, Boston:, 1999.

\bibitem{taha-paper}
M.~El-Taha and S.~Stidham Jr.
\newblock Sample-path stability conditions for multi-server input-output
  processes.
\newblock {\em Journal of Applied Mathematics and Stochastic Analysis}, vol. 7,
  no. 3, pp. 437-456, 1994.

\bibitem{path-mazumdar}
R.~Mazumdar, F.~Guillemin, V.~Badrinath, and R.~Kannurpatti.
\newblock On pathwise behavior of queues.
\newblock {\em Operations Research Letters}, 12:263-270, 1992.

\bibitem{bertsekas-nonlinear}
D.~P. Bertsekas.
\newblock {\em Nonlinear Programming}.
\newblock Athena Scientific, Belmont, MA, 1995.

\bibitem{neely-downlink-ton}
M.~J. Neely, E.~Modiano, and C.~E. Rohrs.
\newblock Power allocation and routing in multi-beam satellites with time
  varying channels.
\newblock {\em IEEE Transactions on Networking}, vol. 11, no. 1, pp. 138-152,
  Feb. 2003.

\bibitem{billingsley}
P.~Billingsley.
\newblock {\em Probability Theory and Measure, 2nd edition}.
\newblock New York: John Wiley \& Sons, 1986.

\end{thebibliography}
\end{document}